\documentclass[11pt]{amsart}
\usepackage{amsmath,amscd,amssymb,amsfonts}
\usepackage{mathrsfs}
\usepackage[sans]{dsfont} % (erase) [sans] for a different style
\usepackage{hyperref}
\numberwithin{equation}{section}       % Number formulas within sections
\numberwithin{figure}{section}       % Number figures within sections
\usepackage{graphicx} % [draft,final]

\theoremstyle{plain}

\newtheorem{prop}{Proposition}[section]

\newtheorem{prob}[prop]{Problem}
\newtheorem{coro}[prop]{Corollary}
\newtheorem{lemm}[prop]{Lemma}

\newtheorem{theoalph}{Theorem}

\theoremstyle{definition}

\theoremstyle{remark}
\newtheorem{rema}[prop]{Remark}

\newtheoremstyle{citing}% name
  {3pt}%      Space above, empty = `usual value'
  {3pt}%      Space below
  {\itshape}% Body font
  {}%         Indent amount (empty = no indent, \parindent = para indent)
  {\bfseries}% Thm head font
  {.}%        Punctuation after thm head
  {.5em}%     Space after thm head: " " = normal interword space;
  {\thmnote{#3}}% Thm head spec

\theoremstyle{citing}
\newtheorem*{generic}{}% all text supplied in the note

\newcommand{\partn}[1]{{\smallskip \noindent \textbf{#1.}}}% for parts of a proof
\DeclareMathAlphabet{\mathpzc}{OT1}{pzc}{m}{it} % Zapf Chancery math alphabet

\newcommand{\N}{\mathbb{N}}

\newcommand{\R}{\mathbb{R}}
\newcommand{\T}{\mathbb{T}}
\newcommand{\Z}{\mathbb{Z}}

\newcommand{\sM}{\mathscr{M}}

\newcommand{\sP}{\mathscr{P}}

\newcommand{\hX}{\widehat{X}}

\newcommand{\hbeta}{\widehat{\beta}}

\newcommand{\hmu}{\widehat{\mu}}

\newcommand{\tX}{\widetilde{X}}

\newcommand{\tDelta}{\widetilde{\Delta}}

\newcommand{\tvarphi}{\widetilde{\varphi}}
\newcommand{\tPhi}{\widetilde{\Phi}}

\newcommand{\dd}{\hspace{1pt}\operatorname{d}\hspace{-1pt}}
\renewcommand{\=}{ : = }

\DeclareMathOperator{\Lip}{Lip} % best Lipschitz constant norm notation
\DeclareMathOperator{\dist}{dist}
\newcommand{\htop}{h_{\operatorname{top}}}
\newcommand{\1}{\mathds{1}}

\DeclareMathOperator{\Leb}{Leb}

\renewcommand{\circle}{\T}

\newcommand{\utheta}{\underline{\theta}}

\newcommand{\signs}{\{ +, - \}^{\N}}
\newcommand{\uvarsigma}{\underline{\varsigma}}

\newcommand{\wtuvarsigma}{\widetilde{\uvarsigma}}

\newcommand{\whm}{\widehat{m}}
\newcommand{\whn}{\widehat{n}}

\newcommand{\Umax}{U_{\max}}
\newcommand{\ellmax}{\ell_{\max}}

\begin{document}

\title[Sensitive dependence of Gibbs measures at low temperatures]{Sensitive dependence of Gibbs measures \\ at low temperatures}
\author{Daniel Coronel}
\address{Daniel Coronel, Departamento de Matem{\'a}ticas, Universidad Andres Bello, Avenida Rep\'ublica~220, Santiago, Chile}
\email{alvaro.coronel@unab.cl}
\author{Juan Rivera-Letelier}
\address{Juan Rivera-Letelier, Facultad de Matem{\'a}ticas, Pontifica Universidad Cat{\'o}lica de Chile, Avenida Vicu{\~n}a Mackenna~4860, Santiago, Chile}
\email{riveraletelier@mat.puc.cl}

\begin{abstract}
The Gibbs measures of an interaction can behave chaotically as the temperature drops to zero.
We observe that for some classical lattice systems there are
interactions exhibiting a related phenomenon of sensitive dependence of Gibbs measures: An
arbitrarily small perturbation of the interaction can produce
significant changes in the low-temperature behavior of its Gibbs measures.

For some one-dimensional $XY$~models we exhibit sensitive dependence of Gibbs
measures for a (nearest-neighbor) interaction given by a smooth function, and for
perturbations that are small in the smooth category.
We also exhibit sensitive dependence of Gibbs measures for an interaction on a classical lattice system with
finite-state space.
This interaction decreases exponentially as a function of the distance
between sites; it is given by a Lipschitz continuous potential in the configuration space.
The perturbations are small in the Lipschitz topology.
As a by-product we solve some problems stated by Chazottes and Hochman.
\end{abstract}

\maketitle

%
%%% Main
%

\section{Introduction}
The Gibbs measures of an interaction can behave chaotically as the temperature drops to zero.
This phenomenon was first exhibited by van~Enter and Ruszel for $N$\nobreakdash-vector models~\cite{vEnRus07}, and later by Chazottes
and Hochman for a classical lattice system with finite-state
space~\cite{ChaHoc10}.
More recently, chaotic temperature dependence was exhibited by the authors for a
quasi-quadratic map and its geometric potential~\cite{CorRivsensitive}.

In~\cite{CorRivsensitive}, a related phenomenon of ``sensitive
dependence of geometric Gibbs measures'' was found: An arbitrarily small
perturbation of the map can produce significant changes in the
low-temperature behavior of its geometric Gibbs measures.
The purpose of this paper is to show that this phenomenon is also present
in some one-dimensional $2$-vector models,\footnote{We thank Aernout van Enter for
  pointing out that the proof of the sensitive dependence of Gibbs
  measures for one-dimensional $2$-vector models given here extends to
  $N$-vector models of arbitrary dimension.} as well as in some classical lattice systems
of arbitrary dimension and finite-state space.
Since the methods used here, based partially on the hats-in-hats idea of~\cite{vEnRus07}, differ significantly from those
of~\cite{CorRivsensitive}, this gives evidence that the sensitive
dependence of Gibbs measures phenomenon is robust; it does not depend
on the particulars of either setting.

Roughly speaking, for a given interaction ``chaotic temperature dependence'' is the non-convergence of
Gibbs measures along a certain sequence of temperatures going to zero.
This concept first arose in the spin-glass literature, where the interactions contain disorder, see for example~\cite{NewSte03a,NewSte07a}.
In contrast, in~\cite{ChaHoc10,vEnRus07,CorRivsensitive} and in this note the interactions are
deterministic and contain no disorder.

In the chaotic temperature dependence, the non-convergence of Gibbs measures
cannot occur for every sequence of temperatures going to zero, due to the compactness of the space of probability measures.
In rough terms, the phenomenon of sensitive dependence of Gibbs
measures exhibited here, is that the non-convergence can indeed occur along any
prescribed sequence of temperatures going to zero, by making an
arbitrarily small perturbation of the original interaction.

To be more precise, we introduce the following terminology.
An interaction~$\Phi$ is \emph{chaotic}, if there is a sequence of
inverse temperatures~$(\beta_\ell)_{\ell \in \N}$ such that~$\beta_\ell \to +
\infty$ as~$\ell \to + \infty$, and such that the following property
holds: If for each~$\ell$ in~$\N$ we choose an arbitrary Gibbs
measure~$\rho_{\ell}$ for the interaction~$\beta_{\ell} \cdot \Phi$,
then the sequence~$(\rho_\ell)_{\ell \in \N}$ does not
converge.\footnote{In this definition, it is crucial that the non-convergence holds for an
arbitrary choice of Gibbs measure for each~$\ell$.
For example, for the $2$-dimensional Ising model and an arbitrary
sequence of inverse temperatures going to infinity, there are some choices
of Gibbs measures for which we have non-convergence, like choosing the ferromagnetic and the
antiferromagnetic phases in an alternate way.
However, there is no sensitive dependence because there
are some choices for which we do have convergence, like choosing
the ferromagnetic phase for each~$\ell$.
See also~\cite{0Hocwebsitenote} for further remarks and clarifications about
this phenomenon.}
That is, an interaction is chaotic if it displays chaotic temperature
dependence in the sense of van Enter and Ruszel~\cite{vEnRus07}.

We also say an interaction~$\Phi$ is \emph{sensitive}, if for
every sequence of inverse temperatures~$(\beta_{\ell})_{\ell \in \N}$ such that~$\beta_\ell \to +
\infty$ as~$\ell \to + \infty$ there is an arbitrarily small
perturbation~$\tPhi$ of~$\Phi$\footnote{The notion of proximity
  between interactions depends on the setting.
For one-dimensional $XY$~models we consider the space of symmetric
nearest-neighbor interactions.
Such an interaction is determined by a function on the circle, as in
the classical $XY$~model.
In this setting we use the smooth topology on the space of smooth
functions on the circle as a notion of proximity between interactions,
see~\S\ref{ss:XY model}.
In the case of a classical lattice system with finite-state space, the
interactions we consider are determined by a potential defined on the
configuration space.
In this setting we use the Lipschitz topology, see~\S\ref{ss:symbolic}.} such that the following property holds:
If for each~$\ell$ in~$\N$ we choose an arbitrary Gibbs
measure~$\rho_{\ell}$ for the interaction~$\beta_{\ell} \cdot \tPhi$,
then the sequence~$(\rho_\ell)_{\ell \in \N}$ does not converge.
Note that in the definition of chaotic interaction the sequence of
inverse temperatures~$(\beta_{\ell})_{\ell \in \N}$ is not arbitrary;
in fact, by the compactness of the space of probability measures we
can always choose a sequence of inverse temperatures and corresponding
Gibbs measures for which we have convergence.
In contrast, in the definition of sensitive interaction the sequence
of inverse temperatures~$(\beta_{\ell})_{\ell \in \N}$ is arbitrary,
but the non-convergence is for Gibbs measures of a perturbation of the
original interaction.\footnote{Note that in the definition of sensitive interaction there is no
assertion about the low-temperature behavior of Gibbs measures of the
original interaction.
In fact, we show that there are sensitive interactions
that are chaotic, and that there are sensitive interactions that are
non-chaotic.}

For one-dimensional $2$-vector (or~$XY$) models, we exhibit a sensitive interaction by modifying the example of van~Enter and
Ruszel, see~\S\ref{ss:XY model}.
In our modification of their example, the (nearest-neighbor)
interaction is determined by a smooth function defined on the circle and the perturbations are
small in the smooth topology.
We show this is in a certain sense best possible: In the analytic
category there is no chaotic interaction, and therefore no sensitive interaction.

In the case of a classical lattice system with finite-state space, we
exhibit a sensitive interaction that decays exponentially as a
function of the distance between sites: The interaction is given by a Lipschitz
continuous potential on the configuration space, and the perturbations
are small in the Lipschitz topology,
see~\S\ref{ss:symbolic}.
We use a new construction that is very flexible and that allows us to solve some of the
problems stated by Chazottes and Hochman
in~\cite{ChaHoc10}.\footnote{It is also possible to modify the
  construction of Chazottes and Hochman in~\cite{ChaHoc10} to exhibit
  a sensitive interaction.
The construction introduced here is more qualitative in nature, and
somewhat simpler.}

\subsection{One-dimensional \texorpdfstring{$XY$}{XY} models}
\label{ss:XY model}
Denote the circle by~$\circle \= \R / \Z$, endowed with the (additive)
group structure inherited from~$\R$.
Given a function~$U \colon \circle \to \R$, consider the nearest-neighbor
interaction~$\Phi_U$ on~$\circle^{\Z}$ defined by
$$ \Phi_U(\{ k, k + 1 \}) \left( \left( \theta_n \right)_{n \in \Z} \right)
\=
- U(\theta_k - \theta_{k + 1}). $$
When~$U$ is continuous there is a unique Gibbs
measure for the interaction~$\Phi_U$, and this measure is translation
invariant, see for example~\cite[Theorem~III$.8.2$]{Sim93} or Lemma~\ref{l:Gibbs measures}.
Denote this measure by~$\rho_U$.

A configuration~$\left( \theta_n \right)_{n \in \Z}$ in~$\circle^{\Z}$
is \emph{ferromagnetic} (resp.~\emph{antiferromagnetic}), if for every~$n$
we have~$\theta_{n + 1} = \theta_n$ (resp.~$\theta_{n + 1} = \theta_n
+ \tfrac{1}{2}$).
Note that a ferromagnetic (resp. antiferromagnetic) configuration is completely
determined by its value at the site~$n = 0$.
The \emph{ferromagnetic} (resp.~\emph{antiferromagnetic}) \emph{phase}
is the measure on~$\circle^{\Z}$ that is evenly distributed on ferromagnetic
(resp. antiferromagnetic) configurations.
\begin{theoalph}[Sensitive dependence of Gibbs measures on the interaction]
\label{t:XY model sensitive dependence}
There is a smooth function~$U_0 \colon \circle \to \R$ such that for every
sequence of positive numbers~$(\hbeta_\ell)_{\ell \in \N}$
satisfying~$\hbeta_{\ell} \to + \infty$ as~$\ell \to + \infty$, the
following property holds: There is an arbitrarily small smooth perturbation~$U$
of~$U_0$ such that the sequence of Gibbs measures~$\left(
  \rho_{\hbeta_\ell \cdot U} \right)_{\ell \in \N}$ accumulates at the
same time on the ferromagnetic and the antiferromagnetic phases.
\end{theoalph}

Using the terminology introduced above, this theorem implies that the
interaction~$\Phi_{U_0}$ is sensitive.
We show~$U_0$ can be chosen so that in addition the
interaction~$\Phi_{U_0}$ is chaotic (resp. non-chaotic), see Remark~\ref{r:sensitive potential}.
More precisely, we show that the function~$U_0$ can be chosen so that
the family of Gibbs measures~$(\rho_{\beta \cdot U_0})_{\beta > 0}$
converges to either the
ferromagnetic, or the antiferromagnetic phase as~$\beta \to + \infty$; it can also be
chosen so that this family of Gibbs measures accumulates at the same time on the ferromagnetic and antiferromagnetic phases as~$\beta \to + \infty$.

The first example of a chaotic interaction was given by van~Enter and Ruszel
in~\cite{vEnRus07} using a discontinuous function~$U$, see also~\cite[\S$6$]{BarCioLopMohSou11}.
We use a modification of their example that allows us to get
smooth functions.
The smooth regularity is essentially optimal: In the real analytic
category there are no chaotic interactions, and therefore no sensitive
ones.
In fact, for every real analytic function~$U \colon \circle \to
\R$, the one-parameter family of Gibbs measures~$(\rho_{\beta \cdot U})_{\beta > 0}$ converges as~$\beta \to + \infty$, see Remark~\ref{r:analytic selection}.
See also~\cite{LopMohSouThi09,BarCioLopMohSou11,LopMenMohSou1210,LopMen14}
for other results on the behavior of Gibbs measures as temperature drops to zero.

The following is our main technical result, from which Theorem~\ref{t:XY model sensitive dependence} follows easily.
Throughout this note we endow~$\{ +, - \}$ with the discrete topology,
and~$\signs$ with the corresponding product topology.
Denote by~$\pi \colon \circle^{\Z} \to \circle$ the projection
defined by
$$ \pi\left( \left( \theta_n \right)_{n \in \Z} \right)
\=
\theta_0 - \theta_1. $$

\begin{generic}[Main Lemma~\ref{t:XY model sensitive dependence}]
There is a family of smooth functions~$\left(  U(\uvarsigma)
\right)_{\uvarsigma \in \signs}$ that is continuous in the~$C^\infty$
topology, and such that the following property holds.
For each integer~$m \ge 1$ put~$\beta_m \= 2^{(m + 10)^3}$.
Then for every~$\uvarsigma = (\varsigma(m))_{m \in \N}$ in~$\signs$, every pair of integers~$\whm$ and~$m$ satisfying
$$ \whm \ge m \ge 1
\text{ and }
\varsigma(m) = \cdots = \varsigma(\whm), $$
and every~$\beta$ in~$\left[ \beta_m, \beta_{\whm} \right]$, the
unique Gibbs measure~$\rho_{\beta \cdot U(\uvarsigma)}$ for the
interaction~$\Phi_{\beta \cdot U(\uvarsigma)}$ satisfies
$$ \rho_{\beta \cdot U(\uvarsigma)} \left( \pi^{-1} \left( \left[ -
      2^{-(m + 1)^2}, 2^{-(m + 1)^2} \right] \right) \right)
\ge 1 - 2^{-m} $$
if~$\varsigma(m) = +$, and
$$ \rho_{\beta \cdot U(\uvarsigma)} \left( \pi^{-1} \left(\left[
      \tfrac{1}{2} - 2^{-(m + 1)^2}, \tfrac{1}{2} + 2^{-(m + 1)^2} \right]
\right) \right)
\ge
1 - 2^{-m} $$
if~$\varsigma(m) = -$.
\end{generic}

The proofs of Theorem~\ref{t:XY model sensitive dependence} and Main Lemma~\ref{t:XY model sensitive dependence} are given in~\S\ref{s:XY model}.

\subsection{Symbolic space}
\label{ss:symbolic}
Let~$d \ge 1$ be an integer, and let~$G$ be either~$\Z^d$ or~$\N_0^d$.  
Given a finite set~$F$ containing at least~$2$ elements, consider the
space~$\Sigma \= F^{G}$ endowed with the distance~$\dist$ defined for distinct elements~$(\theta_n)_{n \in G}$ and~$(\theta_n')_{n \in G}$ of~$\Sigma$, by
$$ \dist \left( (\theta_n)_{n \in G}, (\theta_n')_{n \in G} \right)
\=
2^{- \min \{ \| n \| : \theta_n \neq \theta_n' \}}, $$
where $\|\cdot\|$ is the sup-norm.
Denote by~$\sigma$ the action of~$G$ on~$\Sigma$ by translations,  by~$\sM$ the space of Borel probability measures on~$\Sigma$
endowed with the weak* topology, and by~$\sM_{\sigma}$ the subspace of
those that are invariant by~$\sigma$.
For~$\nu$ in~$\sM_\sigma$, denote by~$h_\nu$ the \emph{measure-theoretic
entropy of~$\nu$}.
The \emph{topological pressure of a continuous function~$\varphi \colon
  \Sigma \to \R$}, is
$$ P(\varphi)
\=
\sup \left\{ h_\nu + \int \varphi \dd \nu : \nu \in \sM_{\sigma}
\right\}. $$
A  \emph{equilibrium state for the potential~$\varphi$} is
a measure~$\nu$ at which the supremum above is attained.   
When~$\varphi$ is Lipschitz continuous and~$F = \Z^d$, the set of
equilibrium states agrees with the set of translation invariant Gibbs
measures, see for instance~\cite[Theorem~$5.3.1$]{Kel98}.
Moreover, if the dimension~$d$ is~$1$, then there is a
unique equilibrium state that we denote by~$\rho_{\varphi}$, see for
example~\cite[Theorem~$1.22$]{Bow08} in the case~$G = \N_0$
and~\cite[Theorem~III.$8.2$]{Sim93} in the case~$G = \Z$.

From now on we use ``translation invariant Gibbs measure''
instead of equilibrium state, even when~$G = \N_0^d$.

\begin{theoalph}[Sensitive dependence of Gibbs measures on the potential]
\label{t:symbolic sensitive dependence}
There is a Lipschitz continuous potential~$\varphi_0 \colon \Sigma \to \R$
and complementary open subsets~$U^+$ and~$U^-$ of~$\Sigma$, such that for every sequence of positive numbers~$(\hbeta_{\ell})_{\ell \in
  \N}$ satisfying~$\hbeta_\ell \to + \infty$ as~$\ell \to + \infty$,
the following property holds: There is an arbitrarily small Lipschitz
continuous perturbation~$\varphi$ of~$\varphi_0$ such that if for
each~$\ell$ we choose an arbitrary translation invariant Gibbs
measure~$\rho_\ell$ for the potential $\hbeta(\ell) \cdot \varphi$,
then the sequence~$(\rho_\ell)_{\ell \in \N}$ accumulates at the same time on a measure supported on~$U^+$ and on a measure supported on~$U^-$.
\end{theoalph}

Using the terminology above, this theorem shows the potential~$\varphi_0$ is sensitive for translation invariant Gibbs
measures.
As for $XY$~models in~\S\ref{ss:XY model}, the interaction can be
chosen to be chaotic, and it can also be chosen to be non-chaotic.

The first examples in a finite-state space of a chaotic interaction
for translation invariant Gibbs measures were given in dimensions~$d = 1$ and~$d \ge 3$ by
Chazottes and Hochman in~\cite{ChaHoc10}.
In dimensions~$d \ge 3$, the interactions constructed by Chazottes and Hochman
are of finite range.
To the best of our knowledge it is open if in dimension~$d = 2$ there
is a finite-range interaction that is chaotic for translation
invariant Gibbs measures.\footnote{We note that for~$d = 2$ the
  interaction given by Theorem~\ref{t:symbolic sensitive dependence}
  is chaotic, but it is not of finite range.
It seems possible to obtain a similar example by adapting the
construction of Chazottes and Hochman in~\cite{ChaHoc10}.
To be more precise, denote by~$X$ the subshift of~$\{0, 1\}^{\N}$
constructed in~\cite{ChaHoc10}, and denote by~$\hX$ an invariant subshift of~$\{0, 1 \}^{\N^2}$ obtained by embedding the
product of countably many copies of~$X$, as in~\S\ref{sss:ground states ge 1}.
Then it seems possible to adapt the computations in~\cite{ChaHoc10},
to show that the potential~$x \mapsto - \dist(x, \hX)$ is chaotic.}
In dimension~$d = 1$, Br{\'e}mont showed that there is no finite-range interaction that is chaotic, see~\cite{Bre03} and also~\cite{Nek04,Lep05,ChaGamUga11}.
See also~\cite{BarLepLop12,Con1307,Lep12,BarLopMen13} and the monograph~\cite{BarLepLop13} for recent related results.

Theorem~\ref{t:symbolic sensitive dependence} follows easily from
the following.
\begin{generic}[Main Lemma~\ref{t:symbolic sensitive dependence}]
There is a continuous family of Lipschitz continuous potentials~$\left( \varphi(\uvarsigma) \right)_{\uvarsigma \in
  \signs}$, complementary open subsets~$U^+$ and~$U^-$ of~$\Sigma$,
and an increasing sequence of positive numbers~$(\beta_m)_{m \in \N}$
converging to~$+ \infty$, such that the following property holds: For every~$\uvarsigma = (\varsigma(m))_{m \in \N}$ in~$\signs$, every pair of integers~$m$ and~$\whm$ such that
$$ \whm \ge m \ge 1
\text{ and }
\varsigma(m) = \cdots = \varsigma(\whm), $$
every~$\beta$ in~$[\beta_m, \beta_{\whm}]$, and every translation invariant Gibbs measure~$\rho$ for the potential $\beta \cdot \varphi(\uvarsigma)$,
we have
$$ \rho(U^{\varsigma(m)})
\ge
1 - 2^{-m}. $$
\end{generic}

The following is a corollary of (the proof of) Main
Lemma~\ref{t:symbolic sensitive dependence}.
This corollary is proven in~\S\ref{ss:ground states}.
A subset~$X$ of~$\Sigma$ is \emph{invariant} if for every~$g$ in~$G$ we
have~$\sigma^g(X) = X$.\footnote{When~$G = \N_0^d$ such a set is
  sometimes called ``forward invariant''.}
\begin{coro}
\label{c:arbitrary accumulation}
Assume that the dimension~$d$ is~$1$.
Let~$X^+$ and~$X^-$ be disjoint compact subsets of~$\Sigma$ that
are invariant, minimal, and uniquely ergodic for~$\sigma$.
Suppose furthermore that~$\htop(\sigma|_{X^+}) =
\htop(\sigma|_{X^-})$, and let~$\rho^+$ and~$\rho^-$ be the unique
invariant probability measure supported on~$X^+$ and~$X^-$, respectively.
Then there is a Lipschitz continuous potential~$\varphi \colon \Sigma \to
\R$ such that the one-parameter family of Gibbs measures~$\left( \rho_{\beta \cdot \varphi} \right)_{\beta > 0}$
accumulates at the same time on~$\rho^+$ and~$\rho^-$ as~$\beta \to + \infty$.
\end{coro}

The hypothesis~$\htop(\sigma|_{X^+}) = \htop(\sigma|_{X^-})$ is
necessary, see~\cite[Corollary~$1$]{AizLie81} or Lemma~\ref{l:ground states maximize entropy}.

Combined with the Jewett-Krieger realization theorem, the following is
a direct consequence of the previous corollary, see for example~\cite[Corollary~$3.2$]{Kri72}.
It solves a problem formulated by Chazottes and
Hochman in~\cite[\S$4.2$]{ChaHoc10}, see also Appendix~\ref{s:marginal}.

\begin{coro}
\label{c:divergence realization}
Assume that the dimension~$d$ is~$1$.
Let~$\mu^+$ and~$\mu^-$ be ergodic measures defined on a Lebesgue
space having the same finite entropy.
Then, provided the finite set~$F$ is sufficiently large, there is a
Lipschitz continuous potential~$\varphi \colon \Sigma \to \R$ such
that the one-parameter family of Gibbs measures~$\left( \rho_{\beta \cdot \varphi} \right)_{\beta > 0}$
accumulates at the same time on a measure isomorphic to~$\mu^+$ and on a measure isomorphic to~$\mu^-$ as~$\beta \to +
\infty$.
\end{coro}

The proof of Main Lemma~\ref{t:symbolic sensitive dependence} is given
in~\S\ref{s:symbolic}.
The deduction of Theorem~\ref{t:symbolic sensitive dependence} from Main Lemma~\ref{t:symbolic sensitive dependence} is analogous to that of
Theorem~\ref{t:XY model sensitive dependence} from Main Lemma~\ref{t:XY model sensitive dependence} given
in~\S\ref{s:XY model}, and we
omit it.

\subsection{Acknowledgments}
\label{ss:acknowledgments}
We would like to thank Aernout van~Enter for useful discussions and
various remarks on an earlier version, Rodrigo Bissacot for useful
discussions, and Ronnie Pavlov for his assistance on multidimensional
subshifts.
We also thank the referees for useful suggestions on improving the presentation.

The first named author acknowledges partial support from FONDECYT grant 11121453, Anillo DYSYRF grant ACT 1103, and Basal-Grant CMM PFB-03.
The second named author acknowledges partial support from FONDECYT grant 1141091.

\section{One-dimensional \texorpdfstring{$XY$}{XY} models}
\label{s:XY model}
The purpose of this section is to prove Theorem~\ref{t:XY model
  sensitive dependence} and Main Lemma~\ref{t:XY model sensitive dependence}.
After some general considerations on Gibbs measures in~\S\ref{ss:Gibbs
measures}, the proofs of these results are given in~\S\ref{ss:proof of XY model chaotic dependence}.

Throughout this section we use~$\Leb$ to denote the probability measure on~$\circle$ induced by the Lebesgue
measure on~$\R$.

\subsection{Gibbs measures of symmetric nearest-neighbor interactions}
\label{ss:Gibbs measures}
Let~$D \colon \circle^{\Z} \to \circle^{\Z}$ be the map defined by
$$ D \left( \left( \theta_n \right)_{n \in \Z} \right)
\=
\left( \left( \theta_n - \theta_{n + 1} \right)_{n \in \Z} \right), $$
and for each~$\theta$ in~$\circle$ denote by~$T_{\theta} \colon
\circle^{\Z} \to \circle^{\Z}$ the map defined by
$$ T_{\theta} \left( \left( \theta_n \right)_{n \in \Z} \right)
\=
\left( \left( \theta_n + \theta \right)_{n \in \Z} \right). $$
Note that for each~$\theta$ in~$\circle$ we have~$D \circ T_{\theta} =
D$, and that for each~$\utheta$ in~$\circle^{\Z}$ we have~$D^{-1}(D(\utheta))
= \{ T_{\theta} (\utheta) : \theta \in \circle \}$.

A measure on~$\circle^{\Z}$ is \emph{symmetric} if for
each~$\theta$ in~$\circle$ it is invariant by~$T_{\theta}$.

\begin{lemm}
\label{l:Gibbs measures}
Let~$U \colon \circle \to \R$ be a continuous function.
Then for every~$\beta$ in~$\R$ there is a unique Gibbs
measure~$\rho_{\beta \cdot U}$ for the interaction~$\Phi_{\beta \cdot
  U}$.
Moreover, $\rho_{\beta \cdot U}$ is characterized as the unique symmetric measure whose
image by~$D$ is equal to
$$ \bigotimes_{\Z} \left(\frac{\exp(\beta \cdot U)}{\int_{\circle} \exp(\beta \cdot
  U(\theta)) \dd \theta} \right) \Leb. $$
In particular, we have
\begin{equation}
\label{e:Markov property}
\pi_* \rho_{\beta \cdot U}
=
\left( \frac{\exp(\beta \cdot U)}{\int_{\circle} \exp(\beta \cdot
  U(\theta)) \dd \theta} \right) \Leb.
\end{equation}
\end{lemm}

The proof of this lemma is given after the following general lemma.
\begin{lemm}
\label{l:symmetric measures}
For every measure~$\mu$ on~$\circle^{\Z}$ there is a unique
symmetric measure~$\hmu$ on~$\circle^{\Z}$ such that~$D_* \hmu = \mu$.
\end{lemm}
\begin{proof}
Given a continuous function~$f \colon \circle^{\Z} \to \R$,
put~$\widehat{f} \= \int f \circ T_{\theta} \dd \Leb(\theta)$, and note
that there is a continuous function~$\widetilde{f} \colon
\circle^{\Z} \to \R$ satisfying~$\widehat{f} = \widetilde{f} \circ D$.

Given a measure~$\mu$ on~$\circle^{\Z}$, the map~$f \mapsto \int
\widetilde{f} \dd \mu$ defines a symmetric measure on~$\circle^{\Z}$ whose image by~$D$ is equal to~$\mu$.
To prove that this is the only measure with these properties, let~$\hmu$ be a symmetric measure on~$\circle^{\Z}$
satisfying~$D_* \hmu = \mu$, and let~$f \colon \circle^{\Z} \to \R$ be a
continuous function.
Then by the change of variable formula we have
$$ \int_{\circle^{\Z}} f \dd \hmu
=
\int_\T \int_{\circle^{\Z}} f \circ T_{\theta} \dd \hmu \dd \theta
=
\int_{\circle^{\Z}} \widehat{f} \dd \hmu
=
\int_{\circle^{\Z}} \widetilde{f} \dd \mu. $$
This proves uniqueness and completes the proof of the lemma.
\end{proof}

\begin{proof}[Proof of Lemma~\ref{l:Gibbs measures}]
Replacing~$U$ by~$\beta \cdot U$ if necessary, assume~$\beta = 1$.

Denote by~$\sP_f(\Z)$ the collection of finite subsets of~$\Z$.
For~$\Lambda$ in~$\sP_f(\Z)$ denote by~$\pi_{\Lambda} \colon \circle^{\Z}
\to \circle^{\Lambda}$ the canonical projection, and by~$\Leb_{\Lambda} \=
\bigotimes_{\Lambda} \Leb$ the product measure on~$\circle^{\Lambda}$.
Moreover, consider the free boundary
condition Hamiltonian~$H_{\Lambda} \colon \circle^{\Lambda} \to \R$ defined by
$$ H_{\Lambda} = \sum_{X \in \sP_f(\Z) : X \subset \Lambda}
\Phi_U(X), $$
and put~$Z_{\Lambda} \= \int_{\circle^{\Lambda}}
\exp \left( -H_\Lambda(\utheta) \right) \dd \Leb_\Lambda(\utheta)$.

For each integer~$n \ge 1$, put
$$ \Lambda_n
\=
\{ -n, \ldots, n \}
\text{ and }
\rho_n
\=
Z_{\Lambda_n}^{-1} \exp \left( -H_{\Lambda_n}\right)
\Leb_{\Lambda_n}. $$
A straightforward computation shows that for every pair of integers~$n$ and~$\whn$ satisfying~$\whn
\ge n \ge 1$, and for every measurable subset~$A$ of~$\circle^{\Lambda_n}$
we have $\rho_{\whn} \left( A \times \circle^{\Lambda_{\whn} \setminus
  \Lambda_n} \right) = \rho_n(A)$.
Thus, by Kolmogorov's theorem there is a unique
measure~$\rho_{\infty}$ on~$\circle^{\Z}$ so that for every integer~$n
\ge 1$ we have~$(\pi_{\Lambda_n})_* \rho_{\infty} = \rho_n$.
The measure~$\rho_\infty$ is clearly translation invariant.

Denote by~$\sM_U$ the simplex of all Gibbs measures for the
interaction~$\Phi_U$.
In part~$1$ below we show that~$\rho_{\infty}$ is absolutely continuous with respect
to each measure in~$\sM_U$, and in part~$2$ we conclude the proof of the lemma using this fact.

\partn{1} 
For each~$\Lambda$ in~$\sP_f(\Z)$ and each~$\utheta'$ in~$\circle^{\Z \setminus \Lambda}$ consider
the Hamiltonian~$H_{\Lambda} \colon \circle^{\Lambda}  \times
\circle^{\Z \setminus \Lambda} \to \R$ defined by
$$ H_{\Lambda}(\utheta | \utheta')
=
\sum_{X \in \sP_f(\Z) : X \cap \Lambda \neq \emptyset} \Phi_U(X)
(\utheta \times \utheta'), $$
and put~$Z_{\Lambda}(\utheta') \= \int_{\circle^{\Lambda}}
\exp \left( -H_\Lambda(\utheta|\utheta') \right) \dd \Leb_\Lambda(\utheta)$.

Putting~$C \= \sup_{\circle} |U|$, which is finite since~$U$ is
continuous, for every integer~$n \ge 1$, every~$\utheta$
in~$\circle^{\Lambda_n}$, and every~$\utheta'$ in~$\circle^{\Z
  \setminus \Lambda_n}$, we have
$$ | H(\utheta|\utheta') - H(\utheta)| \le 2C. $$
It follows that~$Z_{\Lambda_n} \ge \exp(-2C) Z_{\Lambda_n}(\utheta')$
and, together with the DLR equations, that for every~$\rho$ in~$\sM_U$ and every measurable subset~$A$ of~$\circle^{\Lambda_n}$ we have
\begin{multline*}
\rho_{\infty} \left( A \times \circle^{\Z \setminus \Lambda_n} \right)
\\
\begin{aligned}
& =
\int_A Z_{\Lambda_n}^{-1} \exp \left( -H_{\Lambda_n}(\utheta) \right) \dd \Leb_{\Lambda_n}(\utheta)
\\ & =
\int_{\circle^{\Z \setminus \Lambda_n}} \int_A Z_{\Lambda_n}^{-1} \exp \left( -H_{\Lambda_n}(\utheta) \right) \dd \Leb_{\Lambda_n}(\utheta) \dd (\pi_{\Z \setminus \Lambda_n})_* \rho (\utheta')
\\ & \le
\exp(4C) \int_{\circle^{\Z \setminus \Lambda_n}} \int_A Z_{\Lambda_n}(\utheta')^{-1} \exp \left(
  -H_{\Lambda_n}(\utheta|\utheta') \right)
\dd \Leb_{\Lambda_n}(\utheta) \dd (\pi_{\Z \setminus \Lambda_n})_* \rho (\utheta')
\\ & =
\exp(4C) \rho \left( A \times \circle^{\Z \setminus \Lambda_n} \right).
\end{aligned}
\end{multline*}
Since~$n$ and~$A$ are arbitrary, this shows that~$\rho_{\infty}$ is absolutely continuous with respect
to~$\rho$.

\partn{2}
By, e.g., \cite[Corollaries~III$.2.10$ and~III$.3.10$]{Sim93}, there is a translation invariant and ergodic measure~$\rho$ in~$\sM_U$.
By part~$1$ the measure~$\rho_\infty$ is absolutely continuous with
respect to~$\rho$.
Since~$\rho_{\infty}$ is also translation invariant, it follows
that~$\rho = \rho_\infty$, see for
example~\cite[Lemma~$2.2.2$]{Kel98}.
In particular, $\rho_{\infty}$ is in~$\sM_U$.
Let~$\rho'$ and~$\rho''$ be pure states for the interaction~$\Phi_U$,
\emph{i.e.}, extreme points of~$\sM_U$.
By part~$1$ the measure~$\rho_{\infty}$ is absolutely continuous with
respect to~$\rho'$ and with respect to~$\rho''$.
This implies that~$\rho$ and~$\rho'$ are not mutually singular, and
therefore that they are equal, see for
example~\cite[Theorem~III$.5.1$(b)]{Sim93}.
This proves that~$\sM_U$ has a unique extreme point, and therefore
that~$\sM_U$ is reduced to~$\{ \rho_{\infty} \}$.

Since for each~$\theta$ in~$\circle$ the interaction~$\Phi_U$ is
invariant by~$T_{\theta}$, and since~$\rho_{\infty}$ is the
unique Gibbs measure for this interaction, it follows
that~$\rho_{\infty}$ is symmetric.
On the other hand, if for each integer~$n \ge 1$ we denote by~$D_n \colon
\circle^{\Lambda_n} \to \circle^{\Lambda_n \setminus \{ n \}}$ the map
defined by
$$ D_n \left( (\theta_k)_{k \in \Lambda_n} \right)
=
(\theta_k - \theta_{k + 1})_{k \in \Lambda_n \setminus \{ n \}}, $$
then a straightforward computation shows that
\begin{equation*}
\left( \pi_{\Lambda_n \setminus \{ n \}} \right)_* D_* \rho_{\infty}
=
(D_n)_* \rho_n
=
\bigotimes_{\Lambda_n \setminus \{ n \}} \left(
  \frac{\exp(U)}{\int_{\circle} \exp(U(\theta)) \dd\theta} \right) \Leb.
\end{equation*}
Since this holds for every~$n$, this proves~$D_* \rho_{\infty} =
\bigotimes_{\Z} \left( \frac{\exp(U)}{\int_{\circle}
  \exp(U(\theta)) \dd \theta} \right) \Leb$, and together with
Lemma~\ref{l:symmetric measures} concludes the proof of the lemma.
\end{proof}

\subsection{Proofs of Theorem~\ref{t:XY model sensitive
    dependence} and Main Lemma~\ref{t:XY model sensitive dependence}}
\label{ss:proof of XY model chaotic dependence}
Note that, with the notation and terminology in~\S\ref{ss:Gibbs measures}, the
ferromagnetic (resp. antiferromagnetic) phase is characterized as the
unique symmetric measure on~$\circle^{\Z}$ whose image by~$D$ is equal to~$\bigotimes_{\Z}
\delta_0$ (resp.~$\bigotimes_{\Z} \delta_{\tfrac{1}{2}}$).

For an integer~$r \ge 1$ and a function~$\varphi \colon \circle \to \R$
that is~$r$ times continuously differentiable, consider the~$C^r$-norm
of~$\varphi$:
$$ \| \varphi \|_{C^r}
\=
\sum_{\ell = 0}^{r} \left\| \varphi^{(\ell)} \right\|_{\infty}. $$
Fix a smooth function~$\chi \colon \R \to [0, 1]$ that is constant
equal to~$0$ on~$\R \setminus (-1 ,1)$ and constant equal to~$1$
on~$\left[- \tfrac{2}{3}, \tfrac{2}{3} \right]$.
For an interval~$I$ of~$\circle$, denote by~$|I| \= \Leb(I)$ its
length, and let~$\chi_I \colon \circle \to [0, 1]$ be the function that is constant
equal to~$0$ on~$\circle \setminus I$ and that is defined on~$I$ as
follows: Let~$c$ in~$\R$ be such that~$c \mod \Z$ is the middle point
of~$I$, and for each~$x$ in~$\left[ - \tfrac{|I|}{2}, \tfrac{|I|}{2} \right]$ put
$$ \chi_I(c + x \mod \Z) \= \chi \left(\frac{2x}{|I|} \right). $$
Note that for every integer~$\ell \ge 0$ we have~$\left\| \chi_I^{(\ell)}
\right\|_{\infty} = \left( \tfrac{2}{|I|} \right)^{\ell} \left\| \chi^{(\ell)}
\right\|_{\infty}$.
Considering that~$|I| \le 1$, this implies
\begin{equation}
  \label{e:C^r norm}
\| \chi_I \|_{C^r}
\le
\left( \tfrac{2}{|I|} \right)^r \| \chi \|_{C^r}.  
\end{equation}

\begin{proof}[Proof of Main Lemma~\ref{t:XY model sensitive dependence}]
For each integer~$m \ge 0$, define the following intervals of~$\circle$:
$$ I_m^+ \=
\left[ - 2^{-(m + 11)^2}, 2^{-(m + 11)^2} \right]
\text{ and }
I_m^- \= I_m^+ + \tfrac{1}{2}. $$
Moreover, for~$m \ge 1$ put
$$ Y_m^+ = I_{m - 1}^+ \cup I_m^-,
Y_m^- = I_{m - 1}^- \cup I_m^+, $$
$$ \chi_m^+ \= \chi_{I_{m - 1}^+} + \chi_{I_m^-},
\text{ and }
\chi_m^- \= \chi_{I_{m - 1}^-} + \chi_{I_m^+}. $$
Note that for each~$\varsigma$ in~$\{ +, - \}$ we have~$Y_{m + 1}^\varsigma \subset Y_m^\varsigma$.
Finally, put~$\beta_0 \= 2^{10^3}$ and for~$\uvarsigma = \left(
  \varsigma(m) \right)_{m \in \N}$ in~$\signs$ put
$$ U (\uvarsigma)
\=
- \beta_0^{-1} + \sum_{m = 1}^{+ \infty} \left( \beta_{m - 1}^{-1} - \beta_m^{-1} \right)
\chi_m^{\varsigma(m)}. $$

Note that for each integer~$r \ge 1$ we have by~\eqref{e:C^r norm}
\begin{equation*}
  \begin{split}
\sum_{m = 1}^{+ \infty} \beta_{m - 1}^{-1} \left\| \chi_m^{\varsigma(m)} \right\|_{C^r}
& \le
2 \| \chi \|_{C^r} \sum_{m = 1}^{+ \infty} 2^{-(m + 9)^3} \left( \frac{2}{2
    \cdot 2^{-(m + 11)^2}} \right)^r
\\ & \le
2 \| \chi \|_{C^r} \sum_{m = 1}^{+ \infty} 2^{-((m + 9)^3 - r(m + 11)^2)}
\\ & <
+ \infty,  
  \end{split}
\end{equation*}
so the series defining~$U(\uvarsigma)$ converges uniformly with
respect to~$\| \cdot \|_{C^r}$.
It follows that~$U(\uvarsigma)$ is~$r$ times differentiable.
Since~$r \ge 1$ is arbitrary, this proves that~$U(\uvarsigma)$ is smooth.
To prove that~$U(\uvarsigma)$ depends continuously on~$\uvarsigma$
in the~$C^\infty$ topology, let~$r \ge 1$ and~$m_0 \ge 1$ be given
integers and let~$\uvarsigma = (\varsigma(m))_{m \in \N}$
and~$\uvarsigma' = (\varsigma'(m))_{m \in \N}$ be such that for
every~$k$ in~$\{1, \ldots, m_0 \}$ we have~$\varsigma(k) = \varsigma'(k)$.
Then by~\eqref{e:C^r norm} we have
\begin{equation*}
  \begin{split}
\left\| U \left( \uvarsigma \right)
-
U \left( \uvarsigma' \right) \right\|_{C^r}
& \le
\sum_{m = m_0 + 1}^{+ \infty} \beta_{m - 1}^{-1} \left\| \chi_{Y_m^+} - \chi_{Y_m^-} \right\|_{C^r}
\\ & \le
4 \| \chi \|_{C^r} \sum_{m = m_0 + 1}^{+ \infty} 2^{-(m + 9)^3} \left( \frac{2}{2
    \cdot 2^{-(m + 11)^2}} \right)^r
\\ & \le
4 \| \chi \|_{C^r} \sum_{m = m_0 + 1}^{+ \infty} 2^{-((m + 9)^3 - r(m
  + 11)^2)}.
  \end{split}
\end{equation*}
Since this last sum goes to~$0$ as~$m_0 \to + \infty$, it follows that~$U(\uvarsigma)$ depends continuously on~$\uvarsigma$
in the~$C^r$ topology.
Since~$r \ge 1$ is arbitrary, we conclude that~$U(\uvarsigma)$ depends continuously on~$\uvarsigma$ in the~$C^\infty$ topology.

To prove the estimate of the theorem, define for each integer~$m \ge
1$ the subsets of~$\circle$:
$$ M_m^+
\=
\left[ - \tfrac{2}{3} 2^{-(m + 10)^2}, - \tfrac{1}{3} 2^{-(m + 10)^2}
\right] \cup \left[ \tfrac{1}{3} 2^{-(m + 10)^2}, \tfrac{2}{3} 2^{-(m
    + 10)^2} \right] $$
and
$$ M_m^- \= M_m^+ + \tfrac{1}{2}. $$
Note that for each~$\varsigma$ in~$\{ +, - \}$ we have
$$ M_m^\varsigma \subset Y_m^{\varsigma} \setminus \left( Y_{m + 1}^+
  \cup Y_{m + 1}^- \right). $$
Let~$\uvarsigma = (\varsigma(m))_{m \in \N}$ be in~$\signs$ and
let~$\mu$ and~$\hmu$ be in~$\N$ such that
$$ \hmu \ge \mu \ge 1
\text{ and }
\varsigma(\mu) = \cdots = \varsigma(\hmu). $$
Fix~$\beta$ in~$\left[ \beta_{\mu}, \beta_{\hmu} \right]$, let~$m_0$
be the least integer~$m$ in~$\{\mu, \ldots, \hmu \}$ such
that~$\beta_m \ge \beta$, and put
$$ \rho_{\beta} \= \pi_* \rho_{\beta \cdot U(\uvarsigma)}
\text{ and }
Z(\beta) \= \int_{\circle} \exp(\beta \cdot U(\theta)) \dd \theta. $$
By~\eqref{e:Markov property} we have~$\rho_\beta = \frac{\exp(\beta
  \cdot U)}{Z(\beta)} \Leb$.

Noting that~$U(\uvarsigma)$ is constant equal to~$- \beta_{m_0}^{-1}$
on~$M_{m_0}^{\varsigma(m_0)}$, by~\eqref{e:Markov property} we have
\begin{equation}
\label{e:heavy atom}
  \begin{split}
\rho_\beta \left( M_{m_0}^{\varsigma(m_0)}
\right)
& =
\frac{1}{Z(\beta)} \Leb \left( M_{m_0}^{\varsigma(m_0)} \right)
\exp \left( - \beta \beta_{m_0}^{-1} \right)
\\ & =
\left(\tfrac{2}{3} \frac{1}{Z(\beta)} \right)
2^{- (m_0 + 10)^2} \exp \left( - \beta \beta_{m_0}^{-1} \right).
  \end{split}
\end{equation}

On the other hand, noting that~$U(\uvarsigma)$ is constant equal to~$- \beta_0^{-1}$ on~$\circle
\setminus Y_1^{\varsigma(1)}$, by~\eqref{l:Gibbs measures} we have
\begin{equation*}
  \begin{split}
\rho_{\beta} \left( \circle \setminus Y_1^{\varsigma(1)} \right)
& =
\frac{1}{Z(\beta)} \Leb \left( \circle \setminus
  Y_1^{\varsigma(1)} \right) \exp \left( - \beta \beta_0^{-1} \right)
\\ & \le
\frac{1}{Z(\beta)} \exp \left( - \beta \beta_0^{-1} \right).
  \end{split}
\end{equation*}
Combined with~\eqref{e:heavy atom} and~$m_0 \ge 1$, this implies
\begin{equation*}
  \begin{split}
\frac{\rho_{\beta} \left( \circle \setminus Y_1^{\varsigma(1)} \right)}{\rho_{\beta}
  \left( M_{m_0}^{\varsigma(m_0)} \right)}
& \le
\tfrac{3}{2} 2^{(m_0 + 10)^2} \exp \left( - \beta \left(
    \beta_0^{-1} - \beta_{m_0}^{-1} \right) \right)
\\ & \le
\tfrac{3}{2} 2^{(m_0 + 10)^2} \exp \left( - \tfrac{9}{10}
  \beta \beta_0^{-1} \right)
\\ & \le
\tfrac{3}{2} 2^{(m_0 + 10)^2 - \beta \beta_0^{-1}}
\end{split}
\end{equation*}
In the case~$m_0 = 1$ we have~$\beta = \beta_1$, and we obtain an
upper bound of
$$ \tfrac{3}{2} 2^{11^2 - 2^{11^3 - 10^3}}
\le
\tfrac{1}{2^{20}} 2^{-m_0}. $$
In the case~$m_0 \ge 2$ we have~$\beta \ge \beta_{m_0 - 1}$ and
\begin{equation*}
  \begin{split}
\beta \beta_0^{-1} - (m_0 + 10)^2
& \ge
2^{(m_0 + 9)^3 - 10^3} - (m_0 + 10)^2
\\ & \ge
2^{(m_0 - 1) \left( (m_0 + 9)^2 + 100 \right)} - (m_0 + 10)^2
\\ & \ge
2^{100} (m_0 + 9)^2 - (m_0 + 10)^2
\\ & \ge
100 m_0.   
  \end{split}
\end{equation*}
So, in all the cases we obtain
\begin{equation}
\label{e:ground floor}
\frac{\rho_{\beta} \left( \circle \setminus Y_1^{\varsigma(1)} \right)}{\rho_{\beta} \left( M_{m_0}^{\varsigma(m_0)} \right)}
\le
\tfrac{1}{2^{20}} 2^{-m_0}.
\end{equation}

Note that for every~$m$ in~$\N$ we have~$U(\uvarsigma) \le -
\beta_m^{-1}$ on
$$ \Delta_m
\=
Y_m^{\varsigma(m)} \setminus Y_{m + 1}^{\varsigma(m + 1)}. $$
So, by~\eqref{e:Markov property} for every integer~$k \ge - (m_0 -
1)$ we have
\begin{equation*}
  \begin{split}
\rho_\beta \left( \Delta_{m_0 + k} \right)
& \le
\frac{1}{Z(\beta)} \Leb (\Delta_{m_0 + k}) \exp \left( - \beta
  \beta_{m_0 + k}^{-1} \right)
\\ & \le
\left( 4 \frac{1}{Z(\beta)} \right) 2^{-(m_0 + k + 10)^2} \exp \left( - \beta \beta_{m_0 + k}^{-1} \right).
\end{split}
\end{equation*}
Combined with~\eqref{e:heavy atom}, we obtain
\begin{multline}
\label{e:atom relative weight} 
\frac{\rho_{\beta} (\Delta_{m_0 + k})}{\rho_{\beta} \left( M_{m_0}^{\varsigma(m_0)} \right)}
\\ \le
6 \cdot 2^{(m_0 + 10)^2 - (m_0 + k + 10)^2}
\exp \left(- \beta \left( \beta_{m_0 + k}^{- 1} - \beta_{m_0}^{-1} \right) \right).
\end{multline}
In the case~$k \ge 1$ the right-hand side is bounded from above by
$$ (6e) 2^{(m_0 + 10)^2 - (m_0 + k + 10)^2}
\le
(3e) 2^{-2k(m_0 + 10)}
\le
\tfrac{3e}{2^{20}} 2^{-km_0}
\le
\tfrac{1}{2^{16}} 2^{-km_0}. $$
On the other hand, in the case~$k \le - 2$ we have~$m_0 \ge 3$, $\beta
\ge \beta_{m_0 - 1}$,
\begin{equation*}
  \begin{split}
\beta \left( \beta_{m_0 + k}^{-1} - \beta_{m_0}^{-1} \right)
& \ge
\tfrac{9}{10} \beta \beta_{m_0 + k}^{-1}
\\ & \ge
\tfrac{9}{10} 2^{(m_0 + 9)^3 - (m_0 + k + 10)^3}
\\ & =
\tfrac{9}{10} 2^{(|k| - 1) \left( (m_0 + 9)^2 + (m_0 + 9) (m_0 + k + 10) + (m_0 + k + 10)^2 \right)}
\\ & \ge
\tfrac{9}{10} 2^{10|k|(m_0 + 10)},
  \end{split}
\end{equation*}
and therefore
\begin{equation*}
  \begin{split}
\frac{\rho_{\beta} (\Delta_{m_0 + k})}{\rho_{\beta}
  \left( M_{m_0}^{\varsigma(m_0)} \right)}
& \le
6 \cdot 2^{2|k|(m_0 + 10)} \exp \left(- \tfrac{9}{10} 2^{10|k|(m_0 +
    10)} \right)
\\ & \le
6 \cdot 2^{2|k|(m_0 + 10) - 2^{10|k|(m_0 + 10)}}
\\ & \le
6 \cdot 2^{-3|k|(m_0 + 10)}
\\ & \le
\tfrac{1}{2^{20}} 2^{- 3|k| m_0}.
  \end{split}
\end{equation*}
So for every~$k$ different from~$0$ and~$-1$ we have
\begin{equation*}
\frac{\rho_{\beta} (\Delta_{m_0 + k})}{\rho_{\beta}
  \left( M_{m_0}^{\varsigma(m_0)} \right)}
\le
\frac{1}{2^{16}} 2^{-|k|m_0}.
\end{equation*}
Put~$\tDelta_{m_0} = \Delta_1$ if~$m_0 = 1$, and~$\tDelta_{m_0} \=
\Delta_{m_0} \cup \Delta_{m_0 - 1}$ if~$m_0 \ge 2$.
Combined with~\eqref{e:ground floor} and $m_0 \ge 1$, the last estimate implies
\begin{equation}
\label{e:noncentral}
\begin{split}
\frac{\rho_{\beta} (\circle \setminus \tDelta_{m_0})}{\rho_{\beta}
  \left( M_{m_0}^{\varsigma(m_0)} \right)}
& =
\frac{\rho_{\beta} \left( \circle \setminus
      Y_1^{\varsigma(1)} \right)
+
\sum_{\substack{k = - (m_0 - 1) \\ k \neq 0, - 1}}^{+ \infty}
\rho_{\beta} (\Delta_{m_0 +
  k})}{\rho_{\beta} \left( M_{m_0}^{\varsigma(m_0)} \right)}
\\ & \le
\tfrac{1}{2^{20}} 2^{-m_0} + \tfrac{1}{2^{16}} 2 \cdot \sum_{k = 1}^{+ \infty} 2^{-k m_0}
\\ & \le
\tfrac{1}{2^{13}} 2^{- m_0}.
\end{split}
\end{equation}

Since~$U(\uvarsigma) \le - \beta_{m_0}^{-1}$ on~$\Delta_{m_0}$, by~\eqref{e:Markov property} we have
\begin{equation}
\label{e:weak side}
  \begin{split}
\rho_{\beta} \left( \Delta_{m_0} \setminus I_{m_0 - 1}^{\varsigma(m_0)} \right)
& \le
\frac{1}{Z(\beta)} \Leb \left( \Delta_{m_0} \setminus I_{m_0 -
    1}^{\varsigma(m_0)} \right) \exp \left( - \beta
  \beta_{m_0}^{-1} \right)
\\ & \le
\left( 2\frac{1}{Z(\beta)} \right)
2^{-(m_0 + 11)^2} \exp \left( - \beta \beta_{m_0}^{-1} \right).
\end{split}
\end{equation}
Combined with~\eqref{e:heavy atom} this implies
\begin{equation}
\label{e:noncentral atom}
\frac{\rho_{\beta} \left( \Delta_{m_0} \setminus I_{m_0 - 1}^{\varsigma(m_0)} \right)}{\rho_{\beta}
  \left( M_{m_0}^{\varsigma(m_0)} \right)}
\le
3 \cdot 2^{(m_0 + 10)^2 - (m_0 + 11)^2}
\le
\tfrac{1}{2^{19}} 2^{-m_0}.
\end{equation}
If~$m_0 = 1$, then~$\tDelta_{m_0} = \Delta_1$, so the previous
estimate combined with~\eqref{e:noncentral} and the
inclusions~$M_1^{\varsigma(\mu)} \subset I_0^{\varsigma(\mu)}
\setminus I_1^{\varsigma(\mu)} \subset \Delta_1 \cap I_0^{\varsigma(\mu)}$, implies
\begin{equation*}
  \begin{split}
\frac{\rho_{\beta} \left( \circle \setminus \left(
        \Delta_1 \cap I_0^{\varsigma(\mu)} \right) \right)}{\rho_{\beta}
  \left( \Delta_1 \cap I_0^{\varsigma(\mu)} \right)}
& =
\frac{\rho_{\beta} \left( \circle \setminus
    \Delta_1 \right) + \rho_{\beta} \left(
    \Delta_1 \setminus I_0^{\varsigma(\mu)} \right)}{\rho_{\beta}
  \left( \Delta_1 \cap I_0^{\varsigma(\mu)} \right)}
\\ & \le
\tfrac{1}{2^{13}} 2^{-1} + \tfrac{1}{2^{19}} \cdot 2^{-1}
\\ & \le
\tfrac{1}{2^{12}} 2^{-1}.
  \end{split}
\end{equation*}
This proves
$$ \rho_{\beta} \left( I_0^{\varsigma(\mu)} \right)
\ge
\rho_{\beta} \left( \Delta_1 \cap I_0^{\varsigma(\mu)} \right)
\ge
1 - \tfrac{1}{2^{12}} 2^{-1}, $$
and completes the proof of the theorem when~$m_0 = 1$.

It remains to consider the case where~$m_0 \ge 2$.
Suppose~$\varsigma(m_0 - 1) \neq \varsigma(\mu)$.
Then we have~$m_0 = \mu$ and therefore~$\beta = \beta_{m_0}$.
On the other hand, by~\eqref{e:atom relative weight} with~$k = -1$ we have
\begin{equation*}
  \begin{split}
\frac{\rho_{\beta} (\Delta_{m_0 - 1})}{\rho_{\beta}
\left( M_{m_0}^{\varsigma(m_0)} \right)}
& \le
(6e) 2^{(m_0 + 10)^2 - (m_0 + 9)^2} \exp \left(- 2^{(m_0 + 10)^3 - (m_0 + 9)^3} \right)
\\ & \le
(6e) 2^{2m_0 + 19} \exp \left(- 2^{3m_0 + 100} \right)
\\ & \le
6 \cdot 2^{2m_0 + 19} \exp(- (3m_0 + 100))
\\ & \le
6 \cdot 2^{- m_0 - 30}
\\ & \le
\tfrac{1}{2^{20}} 2^{-m_0}.
  \end{split}
\end{equation*}
Combined with~\eqref{e:noncentral}, \eqref{e:noncentral atom}, and the
inclusions
$$ M_{m_0}^{\varsigma(m_0)}
\subset
I_{m_0 - 1}^{\varsigma(\mu)} \setminus I_{m_0}^{\varsigma(\mu)}
\subset
\Delta_{m_0} \cap I_{m_0 - 1}^{\varsigma(\mu)}, $$
we obtain
\begin{equation*}
  \begin{split}
\frac{\rho_{\beta} \left( \circle \setminus \left(
        \Delta_{m_0} \cap I_{m_0 - 1}^{\varsigma(\mu)} \right) \right)}{\rho_{\beta}
  \left( \Delta_{m_0} \cap I_{m_0 - 1}^{\varsigma(\mu)} \right)}
& =
\frac{\rho_{\beta} \left( \circle \setminus \tDelta_{m_0} \right) + \rho_{\beta}
  \left( \Delta_{m_0 - 1} \right)}{\rho_{\beta} \left( \Delta_{m_0} \cap I_{m_0 - 1}^{\varsigma(\mu)} \right)}
\\ & \quad +
\frac{\rho_{\beta} \left( \Delta_{m_0}
    \setminus I_{m_0 - 1}^{\varsigma(\mu)} \right)}{\rho_{\beta}
  \left( \Delta_{m_0} \cap I_{m_0 - 1}^{\varsigma(\mu)} \right)}
\\ & \le
\tfrac{1}{2^{13}} 2^{-m_0} + \tfrac{1}{2^{20}} 2^{-m_0} + \tfrac{1}{2^{19}} 2^{-m_0}
\\ & \le
\tfrac{1}{2^{12}} 2^{-m_0}.
  \end{split}
\end{equation*}
This proves that
$$ \rho_{\beta} \left( I_{m_0 - 1}^{\varsigma(\mu)} \right)
\ge
\rho_{\beta} \left( \Delta_{m_0} \cap I_{m_0 - 1}^{\varsigma(\mu)} \right)
\ge
1 - \tfrac{1}{2^{12}} 2^{-m_0}, $$
and completes the proof of the theorem when~$m_0 \ge 2$
and~$\varsigma(m_0 - 1) \neq \varsigma(\mu)$.

It remains to consider the case where~$m_0 \ge 2$ and~$\varsigma(m_0 -
1) = \varsigma(\mu)$.
By~\eqref{e:Markov property} we obtain, as in~\eqref{e:heavy atom}
and~\eqref{e:weak side} with~$m_0$ replaced by~$m_0 - 1$,
\begin{equation*}
\rho_\beta \left( M_{m_0 - 1}^{\varsigma(m_0 - 1)}
\right)
=
\left(\tfrac{2}{3} \frac{1}{Z(\beta)} \right)
2^{- (m_0 + 9)^2} \exp \left( - \beta \beta_{m_0 - 1}^{-1} \right)
\end{equation*}
and
\begin{equation*}
\rho_{\beta} \left( \Delta_{m_0 - 1} \setminus I_{m_0 - 2}^{\varsigma(\mu)} \right)
\le
\left( 2\frac{1}{Z(\beta)} \right)
2^{-(m_0 + 10)^2} \exp \left( - \beta \beta_{m_0 - 1}^{-1} \right).
\end{equation*}
Therefore,
\begin{equation*}
\frac{\rho_{\beta} \left( \Delta_{m_0 - 1} \setminus I_{m_0 - 2}^{\varsigma(\mu)} \right)}{\rho_\beta \left( M_{m_0 - 1}^{\varsigma(m_0 - 1)}
\right)}
\le
3 \cdot 2^{(m_0 + 9)^2 - (m_0 + 10)^2}
\le
\tfrac{1}{2^{17}} 2^{-m_0}.
\end{equation*}
Combined with~\eqref{e:noncentral}, \eqref{e:noncentral atom}, and the inclusion
$$ M_{m_0}^{\varsigma(\mu)} \cup M_{m_0 - 1}^{\varsigma(\mu)}
\subset
I_{m_0 - 2}^{\varsigma(\mu)} \setminus I_{m_0}^{\varsigma(\mu)}
\subset
\tDelta_{m_0} \cap I_{m_0 - 2}^{\varsigma(\mu)}, $$
we obtain
\begin{equation*}
\begin{split}
\frac{\rho_{\beta} \left( \circle \setminus \left(
        \tDelta_{m_0} \cap I_{m_0 - 2}^{\varsigma(\mu)} \right) \right)}{\rho_{\beta}
  \left( \tDelta_{m_0} \cap I_{m_0 - 2}^{\varsigma(\mu)} \right)}
& =
\frac{\rho_{\beta} \left( \circle \setminus
    \tDelta_{m_0} \right) + \rho_{\beta}
  \left( \Delta_{m_0 - 1}  \setminus I_{m_0 - 2}^{\varsigma(\mu)} \right)}{\rho_{\beta} \left( \tDelta_{m_0} \cap I_{m_0 - 2}^{\varsigma(\mu)} \right)}
\\ & \quad +
\frac{\rho_{\beta} \left( \Delta_{m_0} \setminus I_{m_0 - 1}^{\varsigma(\mu)} \right)}{\rho_{\beta}
  \left( \tDelta_{m_0} \cap I_{m_0 - 2}^{\varsigma(\mu)} \right)}
\\ & \le
\tfrac{1}{2^{13}} 2^{-m_0} + \tfrac{1}{2^{19}} 2^{-m_0} + \tfrac{1}{2^{17}} 2^{-m_0}
\\ & \le
\tfrac{1}{2^{12}} 2^{-m_0}.
  \end{split}
\end{equation*}
This proves
$$ \rho_{\beta} \left( I_{m_0 - 2}^{\varsigma(\mu)} \right)
\ge
\rho_{\beta} \left( \tDelta_{m_0} \cap I_{m_0 - 2}^{\varsigma(\mu)} \right)
\ge
1 - \tfrac{1}{2^{12}} 2^{-m_0}, $$
and completes the proof of the theorem when~$m_0 \ge 2$
and~$\varsigma(m_0 - 1) = \varsigma(\mu)$.
The proof of the lemma is thus complete.
\end{proof}

\begin{proof}[Proof of Theorem~\ref{t:XY model sensitive dependence}]
Fix~$\uvarsigma_0 = (\varsigma_0(m))_{m \in \N}$ in~$\signs$,
put~$U_0 \= U(\uvarsigma_0)$, and fix a integer~$m_0 \ge 1$.
Let~$(\beta_m)_{m \in \N}$ be the sequence given by Main Lemma~\ref{t:XY model sensitive dependence}.
Replacing~$(\hbeta_{\ell})_{\ell \in \N}$ by a subsequence if necessary,
assume this sequence is strictly increasing, that~$\hbeta_1 \ge
\beta_{m_0 + 1}$, and that for every~$m \ge 1$ there is at most~$1$ value
of~$\ell$ such that~$\hbeta_{\ell}$ is in~$[\beta_m, \beta_{m + 1}]$.
For each~$\ell$ in~$\N$ let~$m(\ell)$ be the largest integer~$m \ge 1$ such
that~$\beta_m \le \hbeta_{\ell}$.
Note that~$m(1) \ge m_0 + 1$ and that for every~$\ell$ in~$\N$ the number~$\hbeta_{\ell}$ is
in~$[\beta_{m(\ell)}, \beta_{m(\ell + 1) - 1}]$.

Let~$\uvarsigma = \left( \varsigma(m) \right)_{m \in \N}$ in~$\signs$ be such that
for every~$m$ in~$\{1, \ldots, m_0 \}$ we have~$\varsigma(m) =
\varsigma_0(m)$, and such that for every even (resp. odd) $\ell$ in~$\N$ and every~$m$ in~$[m(\ell),
m(\ell + 1) - 1]$ we have~$\varsigma(m) = +$ (resp.~$\varsigma(m) =
-$).
Then for every~$\ell$ in~$\N$ we have
$$ \varsigma(m(\ell)) = \cdots = \varsigma(m(\ell + 1) - 1), $$
and by Main Lemma~\ref{t:XY model sensitive dependence} we have
$$ \lim_{k \to + \infty} \pi_* \rho_{\hbeta_{2k} \cdot U(\uvarsigma)}
=
\delta_0
\text{ and }
\lim_{k \to + \infty} \pi_* \rho_{\hbeta_{2k + 1} \cdot U(\uvarsigma)}
=
\delta_{\tfrac{1}{2}}. $$
In view of Lemma~\ref{l:Gibbs measures}, this implies that~$\rho_{\hbeta_{2k} \cdot
  U(\uvarsigma)}$ converges to the ferromagnetic phase as~$k \to + \infty$, and
that~$\rho_{\hbeta_{2k + 1} \cdot U(\uvarsigma)}$ converges to the
antiferromagnetic phase as~$k \to + \infty$.
Since~$m_0 \ge 1$ is an arbitrary integer, and since the first~$m_0$
elements of~$\uvarsigma$ and~$\uvarsigma_0$ coincide, it follows
that~$U = U(\uvarsigma)$ can be chosen arbitrarily close to~$U_0$
in the~$C^{\infty}$ topology.
\end{proof}

\begin{rema}
\label{r:sensitive potential}
If in the proof of Theorem~\ref{t:XY model sensitive
  dependence} we choose~$\uvarsigma_0$ as the constant sequence equal
to~$+$ (resp.~$-$), then by Main Lemma~\ref{t:XY model sensitive
  dependence} it follows that the one-parameter family of Gibbs measures~$\left( \rho_{\beta \cdot  U_0} \right)_{\beta > 0}$ converges to the ferromagnetic
(resp. antiferromagnetic) phase as~$\beta \to + \infty$.
In particular, for such~$\uvarsigma_0$ the interaction~$\Phi_{U_0}$ is
non-chaotic.
On the other hand, if we choose~$\uvarsigma_0$ having infinitely
many~$+$'s and infinitely many~$-$'s, then the one-parameter family of
Gibbs measures~$\left( \rho_{\beta \cdot
    U_0} \right)_{\beta > 0}$ accumulates at the same time on the
ferromagnetic and the antiferromagnetic phases as~$\beta \to +
\infty$.
Thus, for such~$\uvarsigma_0$ the interaction~$\Phi_{U_0}$ is chaotic.
\end{rema}

\begin{rema}
\label{r:analytic selection}
In the case~$U$ is real analytic, the interaction~$\Phi_U$ is non-chaotic.
In fact, in this case the one-parameter family of Gibbs measures~$\left( \rho_{\beta \cdot U} \right)_{\beta > 0}$ converges as~$\beta \to + \infty$ to a
measure~$\rho_\infty$ described as follows.
If~$U$ is constant, then for every~$\beta > 0$ we have~$\rho_{\beta
  \cdot U} = \bigotimes_{\Z} \Leb$, and therefore~$\rho_{\infty} = \bigotimes_{\Z} \Leb$.
Assume~$U$ is nonconstant, and note that by Lemma~\ref{l:Gibbs measures} the measure~$\rho_{\infty}$ is symmetric and~$D_* \rho_{\infty}$ is a product measure.
Thus, to describe~$\rho_\infty$ we just need to describe its projection by~$\pi$.
Put~$\Umax \= \sup_{\circle} U$, and for each~$c$ in the finite set~$U^{-1}(\Umax)$ denote by~$\ell(c)$ the least integer~$\ell \ge 1$ such
that~$D^{\ell}U(c) \neq 0$.
Note that~$\ell(c)$ is even and that
$$ \omega(c)
\=
- \frac{D^{\ell(c)}U(c)}{\ell(c)!} > 0. $$
Putting
\begin{equation*}
\ellmax
\=
\max \left\{ \ell(c) : c \in U^{-1}(\Umax) \right\},
\end{equation*}
the measure~$\rho_\infty$ is uniquely determined by
\begin{equation}
\label{e:limit bond}
\pi_* \rho_{\infty}
=
\frac{\sum_{\substack{c \in U^{-1}(\Umax) \\ \ell(c) = \ell_{\max}}}
  \omega(c)^{- \frac{1}{\ellmax}} \delta_c}{\sum_{\substack{c \in
      U^{-1}(\Umax) \\ \ell(c) = \ell_{\max}}}
  \omega(c)^{- \frac{1}{\ellmax}}}.
\end{equation}

To prove this, fix~$\kappa$ in~$(0, 1)$ and note that
$$ \lim_{\beta \to + \infty} \exp(- \beta \Umax) \beta^{\frac{1}{\ellmax}}
\int_{\circle \setminus \bigcup_{c \in U^{-1}(\Umax)} B \left( c, \beta^{-
    \frac{\kappa}{\ell(c)}} \right)} \exp(\beta \cdot U(x)) \dd x
=
0, $$
and that for every~$c$ in~$U^{-1}(\Umax)$
\begin{multline*}
\lim_{\beta \to + \infty} \exp(- \beta \Umax) \beta^{\frac{1}{\ell(c)}}
\int_{B \left( c, \beta^{-\frac{\kappa}{\ell(c)}} \right)} \exp(\beta \cdot U(x)) \dd x
\\ =
\omega(c)^{- \frac{1}{\ell(c)}} \int_{\R} \exp \left( -y^{\ell(c)}
\right) \dd y.
\end{multline*}
These computations prove that
\begin{multline*}
\lim_{\beta \to + \infty} \exp(\beta U_{\max}) \beta^{-
  \frac{1}{\ell_{\max}}} \int_{\circle} \exp(\beta \cdot U(x)) \dd x
\\ =
\left( \int_{\R} \exp \left( -y^{\ell_{\max}} \right) \dd y \right)
\sum_{\substack{c \in U^{-1}(\Umax) \\ \ell(c) = \ell_{\max}}}
\omega(c)^{- \frac{1}{\ellmax}} \delta_c.
\end{multline*}
Combined with~\eqref{e:Markov property}, this implies~\eqref{e:limit bond}.
\end{rema}

\section{Symbolic space}
\label{s:symbolic}
This section is devoted to the proof of Main Lemma~\ref{t:symbolic
  sensitive dependence}.
As mentioned in the introduction, the deduction of
Theorem~\ref{t:symbolic sensitive dependence} from Main Lemma~\ref{t:symbolic sensitive dependence} is analogous to that of
Theorem~\ref{t:XY model sensitive dependence} from Main Lemma~\ref{t:XY model sensitive dependence} given
in~\S\ref{s:XY model}, and we
omit it.

We first prove the following weaker version of Main Lemma~\ref{t:symbolic
  sensitive dependence}, whose proof contains some of the main ideas,
but is simpler.
Corollary~\ref{c:arbitrary accumulation} follows easily from (the proof
of) this result.
\begin{generic}[Main Lemma~\ref{t:symbolic sensitive dependence}']
There is a continuous family of Lipschitz continuous
potentials~$\left( \varphi(\uvarsigma) \right)_{\uvarsigma \in
  \signs}$, complementary open subsets~$U^+$ and~$U^-$ of~$\Sigma$,
and an increasing sequence of positive numbers~$(\beta_m)_{m \in \N}$ converging
to~$+ \infty$, such that the following property holds: For every~$\uvarsigma =
(\varsigma(m))_{m \in \N}$ in~$\signs$ and every integer~$m \ge 1$,
every translation invariant Gibbs measure~$\rho$ for the
potential~$\beta_m \cdot \varphi(\uvarsigma)$ satisfies
$$ \rho(U^{\varsigma(m)})
\ge
\tfrac{2}{3}. $$
\end{generic}

The proof of Main Lemma~\ref{t:symbolic sensitive dependence}' is given
in~\S\ref{ss:proof of symbolic chaotic dependence'}, that of
Corollary~\ref{c:arbitrary accumulation} in~\S\ref{ss:ground states},
and that of Main Lemma~\ref{t:symbolic sensitive dependence} in~\S\ref{ss:proof of symbolic chaotic dependence}.

A subset of~$\Sigma$ is \emph{clopen} if it is at the same time open
and closed.

\subsection{Proof of Main Lemma~\ref{t:symbolic sensitive dependence}'}
\label{ss:proof of symbolic chaotic dependence'}
For a function~$\varphi \colon \Sigma \to \R$, put
$$ \| \varphi \|_{\infty}
\=
\sup \{ |\varphi(x)| : x \in \Sigma \}, $$
$$ | \varphi |_{\Lip}
\=
\sup \left\{ \frac{|\varphi(x) - \varphi(x')|}{\dist(x, x')} : x, x'
  \in \Sigma, x \neq x' \right\}, $$
and
$$\| \varphi \|_{\Lip}
\=
\| \varphi \|_{\infty} + | \varphi |_{\Lip}. $$
A function~$\varphi \colon \Sigma \to \R$ is \emph{Lipschitz continuous} if~$\| \varphi \|_{\Lip} < + \infty$.
Denote by~$\Lip$ the space of all Lipschitz continuous functions.
Then~$\| \cdot \|_{\Lip}$ is a norm on~$\Lip$, for which~$\Lip$ is a Banach space.

The following lemma follows easily from a well-known result, see for
example~\cite[Corollary~$1$]{AizLie81} and also~\cite[Proposition~$29$(ii)]{ConLopThi01} for the case~$G = \N_0$.
We provide the short proof for completeness.

\begin{lemm}
\label{l:ground states maximize entropy}
Let~$X$ and~$X'$ be disjoint compact subsets of~$\Sigma$, each of
which is invariant by~$\sigma$, and such
that~$\htop \left( \sigma|_X \right) > \htop \left( \sigma|_{X'} \right)$.
Moreover, let~$\varphi \colon \Sigma \to \R$ be a Lipschitz continuous
function attaining its maximum precisely on~$X \cup X'$. 
Then for every~$\delta$ in~$(0, 1)$ and every neighborhood~$U$ of~$X$
there is~$\beta_0 > 0$ such that for every~$\beta \ge \beta_0$ and
every translation invariant Gibbs measure~$\rho$ for the potential $ \beta \cdot \varphi$  we have
$$ \rho(U) \ge 1 - \delta. $$
\end{lemm}

In the proof of this lemma, as well as in the proof of
Lemma~\ref{l:uniform dependence of equilibria} below, we use the fact
that the entropy function~$\nu \mapsto h_{\nu}$ is upper
semi-continuous on~$\sM_{\sigma}$, see for example~\cite[Example~$4.2.6$]{Kel98}.

\begin{proof}[Proof of Lemma~\ref{l:ground states maximize entropy}]
It is enough to show that for every family of translation invariant
Gibbs measures~$(\rho_{\beta})_{\beta > 0}$ for the potentials $\beta\cdot \varphi$, every accumulation measure   as~$\beta \to + \infty$ is supported on~$X$.
Let~$(\beta_\ell)_{\ell = 1}^{+ \infty}$ be a sequence of positive
numbers such that~$\beta_\ell \to + \infty$ as~$\ell \to + \infty$ and such
that~$\rho_\ell \= \rho_{\beta_{\ell}}$ converges to a measure~$\rho$
as~$\ell \to + \infty$.
Putting~$m \= \sup \{ \varphi(x) : x \in \Sigma \}$, for
every~$\nu$ in~$\sM_{\sigma}$ that is supported on~$X \cup X'$ we have
$$ \int \varphi \dd \rho
=
\lim_{\ell \to + \infty} \left(
  \frac{h_{\rho_\ell}}{\beta_\ell} + \int \varphi  \dd \rho_{\ell} \right)
\ge
\lim_{\ell \to + \infty}
\left( \frac{h_\nu}{\beta_{\ell}} + \int \varphi \dd \nu \right)
=
\int \varphi \dd \nu = m. $$
It follows that~$\rho$ is supported on~$X \cup X'$.
On the other hand, for every~$\ell \ge 1$ we have
$$ h_{\rho_\ell} + \beta_\ell m
\ge
h_{\rho_\ell} + \beta_\ell \int \varphi \dd \rho_{\ell}
\ge
h_\nu + \beta_{\ell} \int \varphi \dd \nu
=
h_\nu + \beta_{\ell} m, $$
so~$h_{\rho_{\ell}} \ge h_{\nu}$.
Since the entropy function is upper semi-continuous, it follows
that~$h_{\rho} \ge h_\nu$.
Since this holds for every invariant probability measure~$\nu$ supported
on~$X \cup X'$, and by hypothesis~$\htop \left( \sigma|_{X} \right) >
\htop \left( \sigma|_{X'} \right)$, we conclude that~$h_\rho = \htop
\left( \sigma|_{X} \right)$ and that~$\rho$ is supported on~$X$.
\end{proof}

\begin{lemm}
\label{l:uniform dependence of equilibria}
Let~$\varphi_0 \colon \Sigma \to \R$ be a Lipschitz continuous function and
let~$\beta_0 \ge 0$ be given.
Then for every~$\delta > 0$ and every continuous function~$\psi \colon
\Sigma \to \R$ there is~$\varepsilon > 0$ such that for every
Lipschitz continuous function~$\varphi \colon \Sigma \to \R$ satisfying~$\|
\varphi - \varphi_0 \|_{\Lip} \le \varepsilon$ and every translation invariant Gibbs measure~$\rho$ for the potential $\beta_0\cdot \varphi$ there is a translation invariant Gibbs measure~$\nu$ for the potential $\beta_0\cdot \varphi_0$ such that 
$$ \left| \int \psi \dd \rho
-
\int \psi \dd \nu\right|
<
\delta. $$
Moreover, if the dimension~$d$ is~$1$, then for every~$\beta$ in~$[0,
\beta_0]$ we have
$$ \left| \int \psi \dd \rho_{\beta \cdot \varphi}
-
\int \psi \dd \rho_{\beta \cdot \varphi_0} \right|
<
\delta. $$
\end{lemm}
The proof of Lemma \ref{l:continuity of translation invariant Gibbs measures} is after the proof of the following lemma.
\begin{lemm}\label{l:continuity of translation invariant Gibbs measures}
Let~$\varphi \colon \Sigma \to \R$ a Lipschitz continuous function and
$(\varphi_\ell)_{\ell\in \N}$ be a sequence in~$\Lip$ converging
to~$\varphi$.
For every sequence of translation invariant Gibbs measures
$(\rho_\ell)_{\ell\in \N}$ for the potentials in the sequence
$(\varphi_\ell)_{\ell\in \N}$, every accumulation point is an
translation invariant Gibbs measure for the potential $\varphi$.
\end{lemm}
\begin{proof}
To prove this, we use the fact that the pressure~$P(\varphi)$ depends
continuously on~$\varphi$ in~$\Lip$, see for
example~\cite[Theorem~$4.1.10$ b)]{Kel98}.
Let assume that~$\rho_\ell$ converges to a measure~$\rho$
as~$\ell \to + \infty$.
Using that the entropy function is upper semi-continuous
on~$\sM_\sigma$, we have
$$ P(\varphi)
=
\lim_{\ell \to + \infty} P(\varphi_{\ell})
=
\lim_{\ell \to + \infty} \left( h_{\rho_\ell} + \int \varphi_\ell \dd \rho_\ell  \right)
\le
h_{\rho} + \int \varphi \dd \rho
\le
P(\varphi), $$
so~$h_{\rho} + \int \varphi \dd \rho = P(\varphi)$ and therefore~$\rho$ is a translation invariant Gibbs measure for $\varphi$.

\end{proof}
\begin{proof}[Proof of Lemma \ref{l:uniform dependence of equilibria}]
By contradiction. Assume that there are $\delta>0$,  a continuous
function $\psi \colon \Sigma\to \R$, a sequence of Lipschitz continuous
functions $(\varphi_\ell)_{\ell \in \N}$ converging to $\varphi_0$
in~$\Lip$ and a sequence of translation invariant Gibbs measures  $(\rho_\ell)_{\ell \in \N}$  for the sequence of potential $\beta_0\cdot \varphi_\ell$ such that for every translation invariant Gibbs measure $\nu$ for the potential $\beta_0\cdot \varphi_0$ one has 
$$ \left| \int \psi \dd \rho_\ell
-
\int \psi \dd \nu \right|
\ge 
\delta. $$
By compactness the sequence of measures $(\rho_\ell)_{\ell \in \N}$ has an accumulation point which is a translation invariant Gibbs measure for the potential $\beta_0\cdot \varphi_0$ (Lemma \ref{l:continuity of translation invariant Gibbs measures}). This is in contradiction with the previous inequality.

In dimension~$1$, the same argument together with the uniqueness of the
translation invariant Gibbs measures gives the uniformity in $\beta$ in $[0,\beta_0]$. 
\end{proof}

\begin{proof}[Proof of Main Lemma~\ref{t:symbolic sensitive dependence}']
Given a compact subset~$Y$ of~$\Sigma$, denote by~$\chi_Y \colon \Sigma
\to \R$ the function defined by~$\chi_Y(x) \= - \dist(x, Y)$.
A straightforward computation shows that~$\left\| \chi_Y \right\|_{\Lip} \le 2$.

Let~$(X_m^+)_{m \in \N_0}$ and~$(X_m^-)_{m \in \N_0}$ be decreasing
sequences of compact subsets of~$\Sigma$ that are invariant
by~$\sigma$, such that~$X_0^+$ and~$X_0^-$ are disjoint, and such that
for every~$m \ge 0$ we have
\begin{equation}
  \label{e:entropy condition'}
\htop \left( \sigma|_{X_{m + 1}^-} \right)
<
\htop \left( \sigma|_{X_m^+} \right)
\text{ and }
\htop \left( \sigma|_{X_{m + 1}^+} \right)
<
\htop \left( \sigma|_{X_m^-} \right).\footnote{There
  are various choices for~$(X_m^+)_{m \in \N_0}$ and~$(X_m^-)_{m \in
    \N_0}$, see~\S\ref{ss:ground states} for a couple of them.}
\end{equation}
Let~$U^+$ be a clopen neighborhood of~$X_0^+$ that is disjoint
from~$X_0^-$, and put~$U^- \= \Sigma \setminus U^+$.
Finally, for each integer~$m \ge 1$, put
$$ Y_m^+ \= X_{m - 1}^+ \cup X_m^-
\text{ and }
Y_m^- \= X_{m - 1}^- \cup X_m^+, $$
and for~$\varsigma$ in~$\{ +, - \}$ put~$\chi_m^{\varsigma} \=
\1_{\Sigma} + \chi_{Y_m^{\varsigma}}$.

\partn{1}
Define inductively sequences of positive numbers~$(\beta_m)_{m \in \N}$
and~$(\varepsilon_m)_{m \in \N}$, as follows.
Put~$\varepsilon_1 \= 1$, and let~$m \ge 1$ be an
integer such that~$\varepsilon_1$, \ldots, $\varepsilon_m$ are already defined.
In the case~$m \ge 2$, assume that~$\beta_1$, \ldots, $\beta_{m - 1}$
are already defined. 
Given~$\uvarsigma = (\varsigma(k))_{k = 1}^{m}$ in~$\{ +, - \}^{\{1,
  \ldots, m\}}$, put
\begin{equation}
\label{e:finite level potential'}
 \varphi(\uvarsigma)
\=
\tfrac{1}{8} \sum_{k = 1}^m \varepsilon_k \chi_k^{\varsigma(k)},
\end{equation}
and note that~$\varphi(\uvarsigma)$ attains it maximum precisely on~$Y_m^{\varsigma(m)}$.
Let~$\beta(\uvarsigma)$ be the number~$\beta_0$ given by
Lemma~\ref{l:ground states maximize entropy} with~$\varphi =
\varphi(\uvarsigma)$, $X = X_{m - 1}^{\varsigma(m)}$, $X' =
Y_m^{\varsigma(m)} \setminus X_{m - 1}^{\varsigma(m)}$, $\delta =
\tfrac{1}{6}$, and~$U = U^{\varsigma(m)}$.
Define
$$ \beta_m
\=
\max \left\{ \beta(\uvarsigma) : \uvarsigma \in \{ +, - \}^{\{1, \ldots, m\}} \right\}. $$
In the case~$m \ge 2$, replace~$\beta_m$ by~$\beta_{m - 1} + m$ if
necessary, so that~$\beta_{m + 1} > \max \{ \beta_m, m \}$.

To define~$\varepsilon_{m + 1}$, for each~$\uvarsigma = (\varsigma_k)_{k = 1}^m$
in~$\{ +, - \}^{\{1, \ldots, m \}}$ let~$\varepsilon(\uvarsigma)$ be given by Lemma~\ref{l:uniform dependence
  of equilibria} with~$\varphi_0 = \varphi(\uvarsigma)$, $\beta_0 =
\beta_m$, $\delta = \tfrac{1}{6}$, and~$\psi = \1_{U^+}$.
Put
$$ \varepsilon_{m + 1}
\=
\min \left\{ \varepsilon(\uvarsigma) : \uvarsigma \in \{ +, - \}^{\{1,
    \ldots, m\}} \right\}. $$
Replacing~$\varepsilon_{m + 1}$ by~$\varepsilon_m/2$ if necessary,
assume~$\varepsilon_{m + 1} \le \varepsilon_m/2$.

This completes the definition of~$(\varepsilon_m)_{m \in \N}$
and~$(\beta_m)_{m \in \N}$.
Note that for every integer~$m \ge 0$ we have~$\varepsilon_{m + 1} \le
\varepsilon_m/2$ and~$\beta_{m + 1} > \max \{ \beta_m, m \}$, so the sequence~$(\beta_m)_{m \in \N}$ is strictly increasing and~$\beta_m
\to + \infty$ as~$m \to + \infty$.

\partn{2}
Given~$(\varsigma(m))_{m \in \N}$ in~$\signs$, put
$$ \varphi \left( (\varsigma(m))_{m \in \N} \right)
\=
\tfrac{1}{8} \sum_{m = 1}^{+ \infty} \varepsilon_m \chi_m^{\varsigma(m)}. $$
To prove that~$\left( \varphi(\uvarsigma) \right)_{\uvarsigma \in \signs}$
is continuous in~$\Lip$, let~$m_0 \ge 1$ be an integer, and
let~$(\varsigma(m))_{m \in \N}$ and~$(\varsigma'(m))_{m \in \N}$ in~$\signs$ be such that for every~$k$ in~$\{1, \ldots, m_0 \}$ we
have~$\varsigma(k) = \varsigma'(k)$.
Then
\begin{equation*}
  \begin{split}
\left\| \varphi \left( (\varsigma(m))_{m \in \N} \right)
-
\varphi \left( (\varsigma'(m))_{m \in \N} \right) \right\|_{\Lip}
& \le
\tfrac{1}{8} \sum_{m = m_0 + 1}^{+ \infty} \varepsilon_m \left\|
  \chi_{Y_m^+} - \chi_{Y_m^-} \right\|_{\Lip}
\\ & \le
\tfrac{1}{2} \sum_{m = m_0 + 1}^{+ \infty} \varepsilon_m
\\ & \le
\varepsilon_{m_0 + 1}.
  \end{split}
\end{equation*}

To complete the proof of the theorem, let~$\uvarsigma = (\varsigma(k))_{k
  \in \N}$ in~$\signs$ and fix~$m$ in~$\N$.
Put~$\wtuvarsigma \= (\varsigma(k))_{k = 1}^m$ and
let~$\varphi(\wtuvarsigma)$ be defined by~\eqref{e:finite level
  potential'} with~$\uvarsigma$ replaced by~$\wtuvarsigma$.
By our choice of~$\beta_m$, for every~$\beta \ge \beta_m$ and every equilibirum state $\nu$ for the potential $\beta \cdot \varphi(\wtuvarsigma)$ we
have
\begin{equation}
  \label{e:approximated chaoticity'}
\nu
\left( U^{\varsigma(m)} \right)
\ge
\tfrac{5}{6}.
\end{equation}
On the other hand,
\begin{equation*}
  \begin{split}
\| \varphi(\uvarsigma) - \varphi(\wtuvarsigma) \|_{\Lip}
& \le
\tfrac{1}{8} \sum_{k = m + 1}^{+ \infty} \varepsilon_k \left\| \chi_k^{\varsigma(k)}
\right\|_{\Lip}
\\ & \le
\tfrac{1}{8} \sum_{k = m + 1}^{+ \infty} \varepsilon_k \left( 1 + \left\|
    \chi_{Y_k^{\varsigma(k)}} \right\|_{\Lip} \right)
\\ & \le
\tfrac{3}{8} \sum_{k = m + 1}^{+ \infty} \varepsilon_l
\\ & <
\varepsilon_{m + 1},
  \end{split}
\end{equation*}
so by our choice of~$\varepsilon_{m + 1}$ it follows that for
every translation invariant Gibbs measure $\rho$ for the potential $\beta_m\cdot \varphi(\uvarsigma)$ there is an  equilibirum state $\nu$ for the potential $\beta_m \cdot \varphi(\wtuvarsigma)$ such that
$$ \left| \rho(U^{\varsigma(m)})
- \nu(U^{\varsigma(m)}) \right|
<
\tfrac{1}{6}. $$
Together with~\eqref{e:approximated chaoticity'} with~$\beta = \beta_m$
this gives the desired
conclusion.
\end{proof}

\subsection{Ground states}
\label{ss:ground states}
In this section we describe a couple of ways to choose the
sequences~$(X_m^+)_{m \in \N_0}$ and~$(X_m^-)_{m \in \N_0}$ in the proof
of Main Lemma~\ref{t:symbolic sensitive dependence}'.
We discuss the case of dimension~$1$ in~\S\ref{sss:ground states 1},
where we also give the proof of Corollary~\ref{c:arbitrary accumulation}, and the case of
dimension larger than~$1$ in~\S\ref{sss:ground states ge 1}.

\subsubsection{Dimension~$1$}
\label{sss:ground states 1}

We use the following lemma.

\begin{lemm}
\label{l:upper approximation by sft's}
Let~$X$ be a compact subset of~$\Sigma$ that is invariant and transitive for~$\sigma$, and that is not a subshift of finite type.
Then there is a decreasing sequence~$(X_m)_{m \in \N}$ of
compact subsets of~$\Sigma$ that are invariant by~$\sigma$, such that
$$ \bigcap_{m \in \N} X_m = X
\text{ and }
\lim_{m \to + \infty} \htop \left( \sigma|_{X_m} \right) =
\htop \left( \sigma|_{X} \right), $$
and such that for every~$m$ in~$\N$ we have~$\htop \left( \sigma|_{X_{m
    + 1}} \right) < \htop \left( \sigma|_{X_m} \right)$.
\end{lemm}
\begin{proof}
For each integer~$\ell \ge 1$, let~$\Sigma_\ell$ be the subshift of
finite type of all words in~$\Sigma$ with the property that every subword of
length~$\ell$ is a subword of a word in~$X$.
Clearly, $\bigcap_{\ell \in \N} \Sigma_\ell = X$, so~$\lim_{\ell \to + \infty}
\htop(\sigma|_{\Sigma_\ell}) = \htop(\sigma|_X)$, see for example~\cite[Proposition~$4.4.6$]{LinMar95}.
Our hypothesis that~$\sigma$ is transitive on~$X$ implies that for
every~$\ell$ in~$\N$ the map~$\sigma$ is transitive on~$\Sigma_\ell$.
On the other hand, our hypothesis that~$X$ is not a subshift of finite
type implies that for every~$\ell$ there is~$\ell' \ge \ell + 1$ such
that~$\Sigma_{\ell'}$ is strictly contained in~$\Sigma_{\ell}$.
Since~$\Sigma_{\ell}$ and~$\Sigma_{\ell'}$ are both subshifts of finite
type and~$\Sigma_{\ell}$ is transitive, it follows
that~$\htop(\Sigma_{\ell'}) < \htop(\Sigma_{\ell})$, see for example~\cite[Corollary~$4.4.9$]{LinMar95}.
So we can extract a subsequence~$(X_m)_{m \in \N_0}$ of~$\left( \Sigma_{\ell} \right)_{\ell \in \N}$ satisfying the desired properties.
\end{proof}

We now explain a way to choose the sequences~$(X_m^+)_{m \in \N_0}$
and~$(X_m^-)_{m \in \N_0}$ in the proof of Main Lemma~\ref{t:symbolic sensitive dependence}' when the dimension is~$1$.
Let~$X^+$ and~$X^-$ be disjoint and infinite compact subsets
of~$\Sigma$ that are invariant and minimal for~$\sigma$, and
such that~$\htop \left( \sigma|_{X^+} \right) = \htop \left(
  \sigma|_{X^-} \right)$.
Since~$X^+$ (resp.~$X^-$) is infinite and minimal for~$\sigma$, it
follows that it is not a subshift of finite type.
So~$X^+$ and~$X^-$ satisfy the hypotheses of Lemma~\ref{l:upper approximation by sft's}.
Let~$(X_m^+)_{m \in \N_0}$ (resp.~$(X_m^-)_{m \in \N_0}$) be the sequence~$(X_m)_{m \in \N_0}$ given by
Lemma~\ref{l:upper approximation by sft's} with~$X = X^+$ (resp.~$X =
X^-$).
Replacing~$(X_m^+)_{m \in \N_0}$ and~$(X_m^-)_{m \in \N_0}$ by
subsequences if necessary, assume~$X_0^+$ and~$X_0^-$ are disjoint.
These sequences satisfy all the requirements, with the possible
exception of~\eqref{e:entropy condition'} which is easy to satisfy by taking subsequences.

\begin{proof}[Proof of Corollary~\ref{c:arbitrary accumulation}]
In the case~$\rho^+$ (resp.~$\rho^-$) is not purely atomic, the
set~$X^+$ (resp.~$X^-$) is infinite and therefore it is not a subshift of finite type.
So in this case~$X^+$ (resp.~$X^-$) satisfies the hypotheses of Lemma~\ref{l:upper approximation by sft's}.
In the case~$\rho^+$ (resp.~$\rho^-$) is purely atomic, the set~$X^+$
(resp.~$X^-$) is a periodic orbit of~$\sigma$.
We enlarge~$X^+$ (resp.~$X^-$) to a set satisfying the hypotheses of
Lemma~\ref{l:upper approximation by sft's}, as follows.
Suppose first~$G = \Z$, and let~$x_0$ be a point in~$\Sigma$ that is not
in~$X^+$ (resp.~$X^-$) and that differs with some point in~$X^+$
(resp.~$X^-$) only at finitely many positions.
Then the orbit of~$x_0$ is forward and backwards asymptotic to~$X^+$ (resp.~$X^-$),
and the invariant set
$$ X^+ \cup \{ \sigma^g(x_0) : g \in \Z \}, 
\text{ (resp. $X^+ \cup \{ \sigma^g(x_0) : g \in \Z \}$)} $$
is compact and transitive.
Furthermore, this set is not a subshift finite type and~$\rho^+$ (resp.~$\rho^-$) is the
only invariant measure supported on this set.
Suppose now~$G = \N_0$, let~$x_0$ be a point in~$\sigma^{-1}(X^+)
\setminus X^+$ (resp. $\sigma^{-1}(X^-) \setminus X^-$), and
let~$(x_j)_{j \in \N}$ be a sequence in~$\Sigma$ that is asymptotic
to~$X^+$ (resp.~$X^-$) and such that for every~$j$ we
have~$\sigma(x_j) = x_{j - 1}$.
Then the invariant set
$$ X^+ \cup \{ x_j : j \in \N \}
\text{ (resp. $X^- \cup \{ x_j : j \in \N \}$)}$$
is compact and transitive.
Furthermore, this set is not a subshift finite type and~$\rho^+$ (resp.~$\rho^-$) is the
only invariant measure supported on this set.
In all the cases the (enlarged) sets~$X^+$ and~$X^-$ satisfy the hypotheses of
Lemma~\ref{l:upper approximation by sft's}, and~$\rho^+$ and~$\rho^-$
are the only invariant probability measures supported on~$X^+$
and~$X^-$, respectively.

We now explain how to modify the proof of Main Lemma~\ref{t:symbolic sensitive dependence}' to obtain the desired statement.
Choose sequences~$(X_m^+)_{m \in \N_0}$ and~$(X_m^-)_{m \in \N_0}$ as
explained above for our choices of~$X^+$ and~$X^-$.
For each integer~$m$ in~$\N$ choose disjoint clopen neighborhoods~$U_m^+$
and~$U_m^-$ of~$X_m^+$ and~$X_{m - 1}^-$, respectively, so that
$$ \bigcap_{m \in \N_0} U_m^+ = X^+
\text{ and }
\bigcap_{m \in \N_0} U_m^- = X^-. $$
With this notation, the only changes in the proof of Main Lemma~\ref{t:symbolic sensitive dependence}' are the following:
\begin{itemize}
\item
Apply Lemma~\ref{l:ground states maximize entropy} with~$U =
U_m^{\varsigma(m)}$ and~$\delta = 2^{-m}$ instead of~$U =
U^{\varsigma(m)}$ and~$\delta = \tfrac{1}{6}$;
\item
Apply Lemma~\ref{l:uniform dependence of
  equilibria} with~$\delta = 2^{-m}$ instead of~$\delta = \tfrac{1}{6}$.
\end{itemize}
Then we obtain the following: For every~$\left( \varsigma(m) \right)_{m \in \N}$
in~$\signs$, every~$\varsigma$ in~$\{ +, - \}$, and every sequence of
integers~$(m_\ell)_{\ell \in \N}$ in~$\N$ such that~$m_{\ell} \to + \infty$
as~$\ell \to + \infty$ and such that for every~$\ell$ such that~$\varsigma(m_{\ell}) = \varsigma$,
we have
$$ \rho_{\beta_{m_{\ell}} \cdot \varphi(\uvarsigma)} \to
\rho^{\varsigma}
\text{ as }
\ell \to + \infty. $$
So for every~$\left( \varsigma(m) \right)_{m \in \N}$ in~$\signs$
having infinitely many~$+$'s and infinitely many~$-$'s, the Lipschitz
continuous function~$\varphi = \varphi(\uvarsigma)$ satisfies the desired property.
  \end{proof}
\begin{rema}
\label{r:arbitrary accumulation}
By construction, the function~$\varphi(\uvarsigma)$ in the proof of
Corollary~\ref{c:arbitrary accumulation} attains its maximum precisely
on the set~$X^+ \cup X^-$.
It follows that~$\rho^+$ and~$\rho^-$ are the only invariant
and ergodic probability measures~$\rho$ on~$\Sigma$ maximizing the integral~$\int \varphi \dd \rho$.
\end{rema}

\subsubsection{Dimension larger than~$1$}
\label{sss:ground states ge 1}
 We show~$2$ ways to choose the sequences~$(X_m^+)_{m \in \N_0}$ and~$(X_m^-)_{m \in \N_0}$ in the proof
of Main Lemma~\ref{t:symbolic sensitive dependence}'.
The first one is a general way to construct a multidimensional subshift from a one-dimensional subshift, preserving the topological entropy. The second way is based in some results of \cite{Pav11} on multidimensional subshifts of finite type.

We start  with a construction of a multidimensional subshift from a
one-dimensional subshift.
Fix~$d \ge 2$ and recall that $\Sigma=F^{G}$, with $G$ equal to~$\N_0^d$ or~$\Z^d$.
Put~$G_0 \= \N_0$ or~$\Z$, accordingly. 
Given a one-dimensional subshift~$X$ of~$F^{G_0}$, put 
$$
\widehat{X}
\=
\left\{ (x_n)_{n\in G} \in \Sigma : \text{ for all } n_2, \ldots, n_d \in
G_0, \left( x_{(i,n_2,n_3, \ldots, n_d)} \right)_{i \in G_0} \in X
\right\}.
$$
Notice that $\widehat{X}$ can be naturally identified with $X^{G_0^{d
    - 1}}$.
Since $X$ is invariant by translation and closed in the product
topology, the set $\widehat{X}$ is invariant by the action of $G$ and
closed in the product topology.
Thus, $\widehat{X}$ is a multidimensional subshift, and it is easy to
see  that the topological entropies of~$X$ and~$\widehat{X}$ are
equal.
Applying this construction to any pair of sequences  used to prove  Main Lemma~\ref{t:symbolic sensitive dependence}' in dimension~$1$, for example those constructed in the previuos subsection, we obtain the desired sequences in dimension~$d$.

For the second construction, we use the following lemma which follows from some results in
\cite{Pav11}.\footnote{Although the results in~\cite{Pav11} are stated
  for~$G = \Z^d$, they extend without change to the case~$G = \N_0^d$
  by remarking that a subshift of finite type of~$F^{\N_0^d}$ and its
  natural extension have the same topological entropy.
Recall that the natural extension of a subshift of finite type~$X$
of~$F^{\N_0^d}$ is the set of configurations in~$F^{\Z^d}$ having all
of its patterns appearing in some configuration in~$X$.}
For the definition of strongly irreducible subshift we refer the reader to \cite{Pav11}.
 
\begin{lemm} 
\label{l:decreasing entropy subshift}
Let $d\ge 2$.
Let $X$ be a  strongly irreducible subshift of finite type in $\Sigma$
with at least two elements.
Then there is a decreasing sequence $(X_m)_{m\in \N_0}$ of strongly
irreducible subshifts of finite type in~$\Sigma$ starting with~$X$ and
such that for every~$m$ we have~$\htop \left( \sigma|_{X_{m
    + 1}} \right) < \htop \left( \sigma|_{X_m} \right)$.
\end{lemm}
\begin{proof}
Put $X_0=X$.
By \cite[Lemma~$4.11$]{Pav11}, $X_0$ has positive entropy. 
By \cite[Theorem~$1.2$ and Lemma~$9.2$]{Pav11}, if one removes a
sufficiently large pattern of~$X_0$ then the resulting  subshift of
finite type~$X_1$ is strongly irreducible, has at least~$2$ elements,
and its entropy is strictly smaller than the entropy of~$X_0$.
Again by~\cite[Lemma~$4.11$]{Pav11}, $X_1$ has positive entropy. Continuing in this way one can construct a decreasing sequence of strongly  irreducible subshifts of finite type with the desired property. 
\end{proof}

Now we show a way to choose the sequences~$(X_m^+)_{m \in \N_0}$ and~$(X_m^-)_{m \in \N_0}$ in the proof of Main Lemma~\ref{t:symbolic sensitive dependence}' when the dimension~$d$ is larger than~$1$.
Recall that the alphabet~$F$ has at least~$2$ symbols, say~$0$ and~$1$.
Let~$X_0$ be the subshift of finite type of~$\Sigma$ that is contained
in~$\{0, 1\}^G$ and whose set of
forbidden patterns consist of two-site patterns with two consecutive~$1$'s (in each direction, including the diagonals). Clearly, $X_0$ is a
strongly irreducible subshift and by Lemma~\ref{l:decreasing entropy
  subshift} there is a decreasing  sequence of subshifts $(X_m)_{m\in
  \N}$ with strictly decreasing entropy.
For every integer~$m \ge 1$ put $X_m^+ \= X_m$ and let~$X_m^-$ be
the subshift obtained by exchanging~$0$'s and~$1$'s in~$X_m$.
These last two sequences verify the desired properties.

\subsection{Proof of Main Lemma~\ref{t:symbolic sensitive dependence}}
\label{ss:proof of symbolic chaotic dependence}
The following is a variant of Lemma~\ref{l:ground states maximize
  entropy}, with a similar proof.
We include it for completeness.
\begin{lemm}
\label{l:ground states maximize entropy reloaded}
Let~$X$ and~$X'$ (resp.~$\tX$ and~$\tX'$) be disjoint compact subsets
of~$\Sigma$, each of which is  invariant by~$\sigma$, and such
that
$$ \tX \subset X,
\tX' \subset X',
\htop \left( \sigma|_X \right) > \htop \left( \sigma|_{X'} \right),
\text{ and }
\htop \left( \sigma|_{\tX} \right) > \htop \left( \sigma|_{\tX'}
\right). $$
Moreover, let~$\varphi \colon \Sigma \to \R$ and~$\tvarphi \colon \Sigma \to \R$
be Lipschitz continuous functions attaining its maximum precisely
on~$X \cup X'$ and~$\tX \cup \tX'$, respectively.
Then for every~$\varepsilon_0 > 0$, every~$\delta$ in~$(0, 1)$, and
every neighborhood~$U$ of~$X$ there is~$\beta_0 > 0$ such that for
every~$\beta \ge \beta_0$, every~$\varepsilon$ in~$[0,
\varepsilon_0]$ and   every translation invariant Gibbs measure $\rho$ for the potential  $\beta \cdot (\varphi + \varepsilon \tvarphi)$ we have
$$ \rho(U)
\ge
1 - \delta. $$
\end{lemm}
\begin{proof}
It is enough to show that for every sequence~$(\varepsilon_\ell)_{\ell
  \in \N}$ in~$[0, \varepsilon_0]$ and every sequence of positive numbers~$(\beta_\ell)_{\ell \in \N}$ such
that~$\beta_\ell \to + \infty$ and such that every sequence of translation invariant Gibbs measures~$(\rho_\ell)_{\ell \in \N} $ for the potentials  $ 
\beta_\ell \cdot (\varphi +
  \varepsilon_\ell \tvarphi)$ that converges to a measure~$\rho$ as~$\ell
\to + \infty$, the measure~$\rho$ is supported on~$X$.
Note that~$\rho$ is in~$\sM_{\sigma}$.
Taking a subsequence if necessary, assume~$(\varepsilon_\ell)_{\ell
  \in \N}$ converges to a number~$\varepsilon$ in~$[0, \varepsilon_0]$.

Putting~$m \= \sup \{ \varphi(x) + \varepsilon \tvarphi(x) : x
\in \Sigma \}$, for every~$\nu$ in~$\sM_\sigma$ that is supported
on~$\tX \cup \tX'$ we have
\begin{equation*}
  \begin{split}
\int \varphi + \varepsilon \tvarphi \dd \rho
& =
\lim_{\ell \to + \infty} \left( \frac{h_{\rho_\ell}}{\beta_{\ell}} +
    \int \varphi + \varepsilon_{\ell} \tvarphi \dd  \rho_\ell \right)
\\ & \ge
\lim_{\ell \to + \infty} \left( \frac{h_{\nu}}{\beta_{\nu}} + \int
  \varphi + \varepsilon_{\ell} \tvarphi \dd  \nu \right)
\\ & =
\int \varphi + \varepsilon \tvarphi \dd  \nu
\\ & =
m.
\end{split}
\end{equation*}
It follows that~$\rho$ is supported on~$\tX \cup \tX'$ if~$\varepsilon
> 0$, and on~$X \cup X'$ if~$\varepsilon = 0$.
Let~$\nu'$ in~$\sM_{\sigma}$ be supported on~$\tX$ if~$\varepsilon >
0$, and on~$X$ if~$\varepsilon = 0$.
Then for every~$\ell$ we have
$$ h_{\rho_\ell} + \beta_{\ell} m
\ge
h_{\rho_\ell} + \beta_{\ell} \int \varphi + \varepsilon_{\ell}
\tvarphi \dd \rho_{\ell}
\ge
h_{\nu'} + \beta_{\ell} \int \varphi + \varepsilon_{\ell} \tvarphi \dd \nu'
=
h_{\nu'} + \beta_{\ell} m, $$
and therefore~$h_{\rho_\ell} \ge h_{\nu}$.
Since the entropy function is upper semi-continuous, it follows
that~$h_{\rho} \ge h_{\nu'}$.
Since this holds for every~$\nu'$ in~$\sM_{\sigma}$ supported on~$\tX$
if~$\varepsilon > 0$, and on~$X$ if~$\varepsilon = 0$, and since by
hypothesis
$$ \htop \left( \sigma|_{\tX} \right) > \htop \left(
   \sigma|_{\tX'} \right)
\text{ and }
\htop \left( \sigma|_X \right) > \htop \left( \sigma|_{\tX'}
\right),  $$
we conclude that~$h_\rho = \htop \left( \sigma|_{\tX} \right)$
if~$\varepsilon > 0$, and that~$h_{\rho} = \htop \left( \sigma|_X
\right)$ if~$\varepsilon = 0$.
It follows that~$\rho$ is supported on~$\tX$ in the former case, and
on~$X$ in the latter case.
This completes the proof of the lemma.
\end{proof}

\begin{lemm}
\label{l:uniform dependence of equilibria simplified}
Let~$\varphi_0 \colon \Sigma \to \R$ be a Lipschitz continuous function  and
let~$\beta_0 \ge  \beta_0' \ge 0$ be given.
Let~$U$ be a clopen subset of~$\Sigma$ and let~$\delta > 0$ be such
that for every~$\beta$ in~$[\beta_0' , \beta_0]$ and every translation invariant Gibbs measure~$\rho_0$ for the potential~$\beta\cdot \varphi_0$, we have
$$
\rho_0 (U)
\ge
1 - \delta.
$$
Then there is~$\varepsilon > 0$ such that for every Lipschitz
continuous function~$\varphi \colon \Sigma \to \R$ satisfying~$\| \varphi - \varphi_0 \|_{\Lip} \le \varepsilon$, for every~$\beta$
in~$[\beta_0' , \beta_0]$, and every translation invariant Gibbs
measure~$\rho$ for the potential $\beta\cdot \varphi$, we have
$$
\rho (U)
\ge
1 - 2\delta.
$$
\end{lemm}
\begin{proof}
By contradiction.
Let $(\varphi_\ell)_{\ell \in \N}$ be a sequence converging
to~$\varphi_0$ in~$\Lip$, and let $(\beta_\ell)_{\ell\in \N}$ be
a sequence in~$[\beta_0', \beta_0]$, such that the following holds.
For every~$\ell$ in~$\N$ there is a translation invariant Gibbs measure~$\rho_\ell$ for the potential $\beta_{\ell} \cdot \varphi_\ell$ 
such that
$$
\rho_\ell(U)
\le
1- 2\delta.
$$ 
By the compactness of the set of probability measures in the weak*
topology and Lemma~\ref{l:continuity of translation invariant Gibbs
  measures}, there is a translation invariant Gibbs measure~$\rho_0$
for the potential~$\beta \cdot \varphi_0$ such that~$\rho_0(U) \le 1 -
2\delta$, which is a contradiction.
\end{proof}

\begin{proof}[Proof of Main Lemma~\ref{t:symbolic sensitive dependence}]
Let~$(X_m^{\pm})_{m \in \N_0}$, $U^{\pm}$, $Y_m^{\pm}$,
and~$\chi_m^{\pm}$ be as in the proof of Main Lemma~\ref{t:symbolic sensitive dependence}'. 
Define inductively sequences of positive numbers~$(\beta_m)_{m \in \N}$
and~$(\varepsilon_m)_{m \in \N}$ in the same way as in part~$1$ of the
proof of Main Lemma~\ref{t:symbolic sensitive dependence}',
except that~$\beta(\uvarsigma)$ is now defined as the number~$\beta_0$ given by
Lemma~\ref{l:ground states maximize entropy reloaded} with
$$ X = X_{m - 1}^{\varsigma(m)}, X' = Y_m^{\varsigma(m)} \setminus
X_{m - 1}^{\varsigma(m)},
\tX = X_m^{\varsigma(m)}, \tX' = Y_{m + 1}^{\varsigma(m)} \setminus
X_m^{\varsigma(m)}, $$
$$ \varphi = \varphi(\uvarsigma),
\tvarphi = \chi_{m + 1}^{\varsigma(m)},
\varepsilon_0 = \varepsilon_m,
\delta = 2^{-(m + 1)},
\text{ and }
U = U^{\varsigma(m)}, $$
and that~$\varepsilon_{m + 1}$ is defined as follows:
Given~$\varsigma$ in~$\{ +, - \}$, let~$\varepsilon_{m + 1}(\varsigma)$
be the number~$\varepsilon$ given by Lemma~\ref{l:uniform dependence
  of equilibria simplified} with~$U = U^{\varsigma}$, $\delta = 2^{-m}$, and $\beta_0=\beta'_0=\beta_1$ if $m=1$ and $\beta_0=
\beta_{m}$, $\beta'_0=\beta_{m-1}$ if $m\ge 2$, and
put
$$ \varepsilon_{m + 1}
\=
\min \{ \varepsilon_m, \varepsilon_{m + 1}(+), \varepsilon_{m + 1}(-)
\}. $$
For~$\uvarsigma$ in~$\signs$ define~$\varphi(\uvarsigma)$ as in part~$1$ of the
proof of Main Lemma~\ref{t:symbolic sensitive dependence}';
the proof that~$\left( \varphi(\uvarsigma) \right)_{\uvarsigma \in \signs}$
is continuous in~$\Lip$ is the same as that in part~$2$ of Main
Lemma~\ref{t:symbolic sensitive dependence}' and we omit it.

To prove the estimate of the theorem, let~$\uvarsigma = (\varsigma(m))_{m
  \in \N}$ in~$\signs$ be given, let~$m$ and~$\whm$ be integers such that
$$ \whm \ge m \ge 1
\text{ and }
\varsigma(m) = \cdots = \varsigma(\whm), $$
and fix~$\beta$ in~$[\beta_m, \beta_{\whm}]$.
Enlarging~$m$ if necessary, assume~$m$ is the largest integer~$j$ such that~$\beta_j \le \beta$.
The proof in the case~$\beta = \beta_m$ is similar to the proof of Main Lemma~\ref{t:symbolic sensitive dependence}' and we omit
it.
Suppose~$\beta > \beta_m$, and note that~$\beta < \beta_{m + 1}$,
$\whm \ge m + 1$, and~$\varsigma(m + 1) = \varsigma(m)$.
Put~$\wtuvarsigma \= (\varsigma(k))_{k = 1}^{m + 1}$ and let~$\varphi(\wtuvarsigma)$ be defined by~\eqref{e:finite level potential'} with~$\uvarsigma$ replaced by~$\wtuvarsigma$.
By our choice of~$\beta_m$ and our assumption~$\beta \ge \beta_m$, for every translation invariant Gibbs measure~$\rho_0$ for the potential $\beta \cdot \varphi(\wtuvarsigma)$  we
have
\begin{equation*}
\rho_0 (U^{\varsigma(m)})
\ge
1 - 2^{-(m + 1)}.  
\end{equation*}
On the other hand,
\begin{equation*}
  \begin{split}
\| \varphi(\uvarsigma) - \varphi(\wtuvarsigma) \|_{\Lip}
& \le
\tfrac{1}{8} \sum_{k = m + 2}^{+ \infty} \varepsilon_k \left\| \chi_k^{\varsigma(k)}
\right\|_{\Lip}
\\ & \le
\tfrac{1}{8} \sum_{k = m + 2}^{+ \infty} \varepsilon_k \left( 1 + \left\|
    \chi_{Y_k^{\varsigma(k)}} \right\|_{\Lip} \right)
\\ & \le
\tfrac{3}{8} \sum_{k = m + 2}^{+ \infty} \varepsilon_l
\\ & <
\varepsilon_{m + 2},
  \end{split}
\end{equation*}
so by our choice of~$\varepsilon_{m + 2}$ and the inequality~$\beta <
\beta_{m + 1}$, for every translation invariant Gibbs measure~$\rho$
for the potential~$\beta \cdot \varphi(\uvarsigma)$ we have
\begin{equation*}
\rho (U^{\varsigma(m)})
\ge
1 - 2^{-m}.
\end{equation*}
This completes the proof of the lemma.
\end{proof}

\appendix
\section{Zero-temperature convergence and marginal entropy}
\label{s:marginal}
In their example of a Lipschitz continuous potential whose Gibbs measures do not converge
as the temperature goes to zero,
Chazottes and Hochman considered a minimal set supporting~$2$ distinct
ergodic probability measures, see~\cite[\S\S$3$, $4.1$]{ChaHoc10}.
As explained in~\S$4.3$ of that paper, a key property of their
example is that at certain scales the marginal entropies of these measures
are sufficiently different.
They asked whether such a connection between the convergence of Gibbs
measures at zero temperature and marginal entropies exists in general, see Problem~\ref{p:marginal} below for a precise formulation.
The purpose of this appendix is to exhibit examples for which this is not
the case, thus answering the question of Chazottes and Hochman in the
negative.

To formulate the question of Chazottes and Hochman more precisely,
put~$F \= \{0, 1 \}$, $\Sigma = \{0, 1\}^{\N_0}$, and let~$\varphi \colon \Sigma \to
\R$ be a Lipschitz continuous potential.
Denote by~$\sM_{\sigma}(\varphi)$ the space of invariant probability measures~$\rho$ on~$\Sigma$ that are invariant
by~$\sigma$ and that maximize~$\int \varphi \dd \rho$.
For each integer~$n \ge 1$ let~$\sM_n^*$ be the set of marginal
distributions obtained by restricting a measure in~$\sM_{\sigma}(\varphi)$
to~$\{0, 1\}^n$, \emph{i.e.}, if we identify~$\{0, 1 \}^{\{0, \ldots, n - 1 \}}$
  with~$\{0, 1\}^n$ and denote by~$\pi_n \colon \Sigma \to \{0, 1 \}^n$
  canonical projection, then
$$ \sM_n^*
=
\left\{ (\pi_n)_* \mu : \mu \in \sM_{\sigma}(\varphi) \right\}. $$
Since the entropy function defined on~$\sM_n^*$ is strictly concave,
it attains its maximum at a unique point~$\mu_n^*$; put
$$ \sM_n
\=
\left\{ \mu \in \sM_{\sigma}(\varphi) : (\pi_n)_*\mu = \mu_n^* \right\}. $$
Note that if for each~$n$ we choose a measure~$\mu_n$ in~$\sM_n$, then
the set of accumulation measures of the sequence~$\left( (\pi_n)_* \mu_n \right)_{n \in \N}$ is independent of~$\left( \mu_n \right)_{n \in \N}$.
In the case the limit~$\lim_{n \to + \infty} (\pi_n)_* \mu_n$ exists,
we denote it by~$\lim_{n \to + \infty} \sM_n$.
\begin{prob}[\cite{ChaHoc10}, \S$4.3$]
\label{p:marginal}
Is the existence of~$\lim_{\beta \to + \infty}
\rho_{\beta \cdot \varphi}$ equivalent to the existence of~$\lim_{n \to + \infty} \sM_n$?
\end{prob}

Denote by~$\overline{0}$ (resp.~$\overline{1}$) the constant
sequence in~$\Sigma$ equal to~$0$ (resp.~$1$), and let~$\varphi$ be the Lipschitz continuous function given by (the proof of) Corollary~\ref{c:arbitrary accumulation} with~$X^+ = \{
\overline{0} \}$ and~$X^- = \{ \overline{1} \}$.
Then the Gibbs measures~$\left( \rho_{\beta \cdot \varphi}
\right)_{\beta > 0}$ accumulate at the same time
on~$\delta_{\overline{0}}$ and on~$\delta_{\overline{1}}$ as~$\beta
\to + \infty$, so the limit~$\lim_{\beta \to + \infty} \rho_{\beta
  \cdot \varphi}$ does not exist.
We show below that~$\lim_{n \to + \infty} \sM_n$ exists.
This answers negatively the question in Problem~\ref{p:marginal}.

To prove that~$\lim_{n \to + \infty} \sM_n$ exists, we use that
\begin{equation}
  \label{e:fixed are maximal}
\sM_{\sigma}(\varphi)
=
\left\{ \alpha_0 \delta_{\overline{0}} + \alpha_1 \delta_{\overline{1}} : \alpha_0 \ge 0, \alpha_1 \ge 0, \alpha_0 + \alpha_1 = 1 \right\}, \end{equation}
see Remark~\ref{r:arbitrary accumulation}.
Thus, if we denote by~$\overline{0}_n$ (resp.~$\overline{1}_n$) the
constant sequence in~$\{0, 1 \}^n$ equal to~$0$ (resp.~$1$), then we have
$$ \sM_n^*
=
\left\{ \alpha_0 \delta_{\overline{0}_n} + \alpha_1 \delta_{\overline{1}_n} : \alpha_0 \ge 0, \alpha_1 \ge 0, \alpha_0 + \alpha_1 = 1 \right\}, $$
and therefore
$$ \mu_n^* = \tfrac{1}{2} \left(\delta_{\overline{0}_n} +
\delta_{\overline{1}_n} \right)
\text{ and }
\sM_n = \left\{ \tfrac{1}{2} \left( \delta_{\overline{0}} +
    \delta_{\overline{1}} \right) \right\}. $$
It follows that~$\lim_{n \to + \infty} \sM_n = \tfrac{1}{2} \left( \delta_{\overline{0}} +
\delta_{\overline{1}} \right)$.

% For Bibtex bibliography:
\bibliographystyle{alpha}
%\bibliography{$HOME/papers/0bib/papers}

\end{document}